\documentclass[letterpaper]{article}
\usepackage{aaai2026}
\usepackage{times}
\usepackage{helvet}
\usepackage{courier}
\usepackage[hyphens]{url}
\usepackage{graphicx}
\usepackage{natbib}
\usepackage{caption}
\usepackage{amsmath,amssymb}

\usepackage{microtype}
\usepackage{xcolor}
\usepackage{siunitx}
\usepackage{algpseudocode}
\usepackage{algpseudocodex}
\usepackage{algorithm}
\renewcommand{\algorithmicrequire}{\textbf{Input:}}
\renewcommand{\algorithmicensure}{\textbf{Output:}}
\algrenewcommand\algorithmicindent{1em}
\usepackage{amsthm}
\usepackage{mathtools}
\usepackage{cleveref}
\crefname{table}{Table}{Tables}
\crefname{prop}{Proposition}{Proposition}
\usepackage{makecell}
\usepackage{appendix}
\usepackage{tikz}
\usepackage{thm-restate}
\usetikzlibrary{decorations.pathreplacing}
\usetikzlibrary{positioning,shapes.geometric}
\usetikzlibrary{arrows.meta}
\usetikzlibrary{intersections}
\theoremstyle{plain}
\newtheorem{thm}{Theorem}

\newtheorem{prop}[thm]{Proposition}
\newtheorem{cor}[thm]{Corollary}

\RequirePackage{xspace}
\newcommand{\n}[1]{{#1}\xspace}  
\DeclareMathOperator*{\argmax}{arg\,max}
\DeclareMathOperator*{\argmin}{arg\,min}
\DeclareMathOperator*{\dist}{\textsf{dist}}
\DeclareMathOperator*{\diam}{\textsf{diam}}
\newcommand{\algname}{\n{PULL}}
\newcommand{\procname}{\n{\texttt{PULL}}}
\newcommand{\prbname}{\n{\textsc{CUMAPF}}}
\newcommand{\mapf}{\n{\textsc{MAPF}}}

\newcommand{\lacam}{\n{LaCAM$^*$}}
\newcommand{\ilp}{\n{ILP}}

\newcommand{\indG}[1]{G[#1]}
\newcommand{\nei}[1]{N(#1)}
\newcommand{\neicl}[1]{N[#1]}
\newcommand{\degv}[1]{d(#1)}

\newcommand{\reachable}{F}
\newcommand{\reachablek}[1]{F_{#1}}

\newcommand{\qfrom}{\mathcal{Q}^{\text{from}}}
\newcommand{\setqfrom}{\qfrom}
\newcommand{\current}{\mathcal{Q}^{\text{to}}}
\newcommand{\setcurrent}{\current}
\newcommand{\currentk}[1]{\current_{#1}}

\newcommand{\conf}[1]{\mathcal{Q}_{#1}}
\newcommand{\setconf}[1]{\conf{#1}}

\newcommand{\cutset}{B}
\newcommand{\cutsetk}[1]{B_{#1}}

\newcommand{\reserved}{R}
\newcommand{\reservedk}[1]{R_{#1}}
\newcommand{\finalk}[1]{t_{#1}}

\newcommand{\bfs}[2]{\textsf{Reach}(#1,#2)}

\usepackage{newfloat}
\usepackage{listings}
\usepackage{booktabs,makecell,threeparttable,caption}
\usepackage{adjustbox}

\newcommand{\valci}[3]{\makecell{\footnotesize #1\\[-1pt]\scriptsize (#2\,--\,#3)}}
\newcommand{\mth}[2]{\makecell{#1\\[-1pt]\scriptsize #2}}
\newcommand{\mpar}[2]{#1~~{\footnotesize \textcolor{darkgray}{$\pm$#2}}}

\usepackage{subcaption}
\DeclareCaptionStyle{ruled}{labelfont=normalfont,labelsep=colon,strut=off} 
\floatstyle{ruled}
\newfloat{listing}{tb}{lst}{}
\floatname{listing}{Listing}

\title{Polynomial-time Configuration Generator for\\Connected Unlabeled Multi-Agent Pathfinding}
\author {
Takahiro Suzuki\textsuperscript{\rm 1,2},
Keisuke Okumura\textsuperscript{\rm 1,3}
}
\affiliations {
\textsuperscript{\rm 1}National Institute of Advanced Industrial Science and Technology (AIST), Japan\\
\textsuperscript{\rm 2}Tohoku University, Japan\\
\textsuperscript{\rm 3}University of Cambridge, UK\\
takahiro.suzuki.q4@dc.tohoku.ac.jp, okumura.k@aist.go.jp
}
\nocopyright

\begin{document}

\maketitle

\begin{abstract}
We consider \emph{Connected Unlabeled Multi-Agent Pathfinding} (\prbname), a variant of \mapf where interchangeable agents must be connected at all times. 
This problem is fundamental to swarm robotics applications such as self-reconfiguration and marching, where standard \mapf is insufficient as it does not guarantee the connectivity constraint.
Despite its simple structure, \prbname remains understudied and lacks practical algorithms.
We first develop an \textsc{Integer Linear Programming} (ILP) reduction to solve \prbname. 
Although this formulation provides a makespan-optimal plan, it is severely limited in terms of scalability and real-time responsiveness due to the large number of variables.
We therefore propose a suboptimal but complete algorithm named \emph{\algname}.
It is based on a rule-based one-step function that computes a subsequent configuration that preserves connectivity and advances towards the target configuration.
\algname is lightweight, and runs in $O(n^2)$ time per step in 2D grid, where $n$ is the number of agents. 
Empirically, \algname can quickly solve randomly generated instances containing hundreds of agents, which ILP cannot handle.
Furthermore, \algname's solution substantially improves upon a naive approach to CUMAPF.
\end{abstract}

\section{Introduction}
Modern robotics is increasingly tackling complex challenges that require tight coordination among multiple agents, where effective collaboration frequently proves essential for task success. This is especially apparent in demanding scenarios such as collaborative work in space or autonomous exploration at disaster sites, where maintaining the integrity of the group becomes a critical factor.

This motivates the study of \emph{Connected Unlabeled Multi-agent Pathfinding} (\prbname): the problem of computing collision-free paths for a team of robots to reach target placements while always maintaining geometric connectivity. The importance of this problem is underscored by its applications in autonomous platooning~\cite{Diaz21}, swarm navigation~\cite{LF25}, and programmable matter~\cite{Hinnenthal24, Kostitsyna23}. 

\prbname can be viewed as a variant of the broader research field of \emph{Multi-Agent Pathfinding} (\mapf), which is a general term for the problem of planning paths for multiple agents in a shared environment that do not conflict with each other.
\mapf problems can be classified into \emph{labeled}, \emph{unlabeled}, and \emph{colored} settings~\cite{Stern19}, 
and \prbname can be viewed in the context of the second setting.
This setting can be applied when all agents can exchange roles or when they are functionally identical.

While unlabeled \mapf can be optimized in polynomial time for general cases~\cite{Yu13unlabeled}, adding connectivity constraints makes the problem significantly harder. 
Indeed, \citet{Fekete21} showed that makespan optimization for \prbname is NP-hard even on simple 2D grids.
Moreover, unlike the unlabeled case, various algorithms developed for labeled \mapf (see \Cref{tab:MAPFwork}) do not work on \prbname, 
since they cannot explicitly handle connectivity constraints.
This suggests that \prbname raises novel challenges unaddressed by existing frameworks, necessitating a new algorithmic framework.

In the context of MAPF, various algorithms have been proposed depending on the application and the characteristics of the solution needed: e.g., 
computationally intensive
\emph{optimal algorithms}, and lightweight \emph{suboptimal algorithms}.
Optimal algorithms, like CBS~\cite{Sharon15}, guarantee minimum solution length; however, their computational cost is prohibitive for large-scale systems. 
In contrast, suboptimal algorithms do not guarantee optimality, but they provide fast output even for large-scale instances~\cite{Li22,PIBT22,LaCAM23}.
Research on large-scale automation has become increasingly important, demanding a suboptimal algorithm for practical applications. 
\Cref{tab:MAPFwork} shows several related works of MAPF and its variants.
Although there are many algorithms for large-scale labeled and unlabeled MAPF, their counterparts for \prbname are missing; at present, we only know the NP-hardness, and an approximation algorithm that runs in restricted conditions.

\begin{table*}
    \centering
    \setlength{\tabcolsep}{5pt}
    \setcellgapes{-2pt}
\makegapedcells
{\small
\begin{tabular}
    {@{}llll@{}}
    \toprule
       & optimal algorithm & suboptimal algorithm & theoretical work \\
      \midrule
     Labeled \mapf & \makecell[l]{M$^*$~\cite{Wagner15},\\CBS~\cite{Sharon15}~
     $\ldots$} & 
     \makecell[l]{
     PIBT~\cite{PIBT22},\\
     LNS2~\cite{Li22}~
     $\ldots$
     }&
     \makecell[l]{
     Poly-time for existence~\cite{Roger12},\\
     NP-hard for optimization~\cite{Yu13hard}\\
     }\\
     \midrule
     Unlabeled \mapf & 
     \makecell[l]{
      Reduction to \textsc{Max-Flow}\\
      \cite{Yu13unlabeled}
     } & 
     \makecell[l]{
      TSWAP\\
      \cite{TSWAP23}
     }& 
     \makecell[l]{
     Poly-time for optimization~\cite{Yu13unlabeled}\\
     }\\
     \midrule
     \prbname & \makecell[l]{
        \textbf{Reduction to ILP (this work)},\\
        \textcolor{lightgray}{(\algname + \lacam; appendix)}
     } & 
     \makecell[l]{
        \textbf{\algname (this work)}
     }& 
     \makecell[l]{
        NP-hard for optimization
        \cite{Fekete21}
     }\\
     \bottomrule
\end{tabular}
}

\caption{Categorization of \mapf variants and their solutions. 
While we focus on ILP and \algname for \prbname, \Cref{sec:lacam} further explores an eventually optimal algorithm, which quickly finds initial suboptimal solutions and refines them towards the optimum.
This is achieved by embedding \algname into a configuration generator-based \mapf algorithm called \lacam~\cite{LaCAM23}.
We omitted it due to limited page space and the subtle practical effects of solution refinement; see \Cref{sec:lacam,sec:explacam} for the discussion.}

\label{tab:MAPFwork}
\end{table*}
\paragraph{Contribution.} 
In this paper, we provide both optimal and suboptimal approaches for \prbname.
We first present a reduction-based algorithm that utilizes an \textsc{Integer Linear Programming} (\ilp).
Although this method finds a makespan-optimal solution, it is not scalable and is unsuitable for real-time applications due to the large number of variables and constraints.
This motivates us to develop a suboptimal yet polynomial-time and complete algorithm.
The proposed \emph{\algname} algorithm is a \emph{configuration generator}, efficiently computing the next placement of the agents from the current ones and the target.
\algname is computationally lightweight; it runs in $O(\Delta^2n^2)$ time per one step, where $\Delta$ and $n$ are the maximum degree of a graph and the number of agents, respectively.
Furthermore, \algname outputs empirically efficient solutions compared to those obtained by a trivial approach (e.g., 350\% improvement for 500 agents on a 32$\times$32 grid).
In what follows, we present problem formulation, related work, algorithms, and evaluation in order.
The code is available on \url{https://github.com/su3taka/CUMAPF-PULL}.

\section{Preliminaries}\label{sec: preliminaries}

\paragraph{Notation.}
Let $G = (V, E)$ be a simple, undirected, unweighted, and finite connected graph. 
For a vertex $v \in V$, the \emph{open neighborhood} is denoted by $\nei{v} = \{w \in V \mid vw \in E\}$, and its \emph{degree} is $\degv{v} = |\nei{v}|$. 
The maximum degree of $G$ is $\Delta = \max_{v \in V} \degv{v}$. 
For a vertex set $V' \subseteq V$, their neighborhoods are $N(V') = (\bigcup_{v \in V'} N(v)) \setminus V'$.
For a vertex subset $S \subseteq V$, we denote the subgraph induced by $S$ as $\indG{S}$. 
A graph is \emph{connected} if there is a path between any two distinct vertices. 
We use a notation for a set of continuous integers: let $[i,j]$ be a set $\{i,i+1,\dots, j-1,j\}$ for $i\le j$.

Let $\bfs{G}{v} \subseteq V$ be the set of all vertices reachable from $v$ in $G$. 
This set can be computed in linear time by using breadth-first search (BFS), simultaneously computing their shortest path from $u\in \bfs{G}{u}$ to $v$ and its length $\dist(u,v)$.
A \emph{component} of $G$ is a subgraph of the form $G[\bfs{G}{v}]$
for some $v \in V$.
A vertex $v \in V$ is a \emph{cut vertex} if $\indG{V \setminus \{v\}}$ has more components than $G$. 
The set of all cut vertices of $G$ is denoted by $\textsf{Cut}(G)$. This set can be computed in linear time by using depth-first search (DFS).

\paragraph{Problem Definition.}
Let $G = (V, E)$ be a graph and $A = [1,n]$ be a set of $n$ agents. 
A \emph{configuration} $\mathcal{Q} = (q_1, \dots, q_n) \in V^n$ is an assignment of each agent to a vertex. 
Although the configuration $\mathcal{Q}$ is defined as a list of vertices, we may also refer to it as a set of vertices for convenience.
For the sake of simplicity, let $\mathcal{Q}(R)$ be $\{\mathcal{Q}[i] \mid i \in R\}$ for a set of agents $R\subseteq A$. 
A configuration $\mathcal{Q}$ is \emph{connected} if the induced subgraph $\indG{\mathcal{Q}}$ is connected.

Given two sets of vertices $S, T \subset V$ with $|S|=|T|=n$, the goal of \prbname is to find a finite sequence of configurations (a \emph{plan}) $\Pi = [\conf{0}=S, \conf{1}, \ldots, \conf{t^{\ast}}=T]$. 
Throughout the plan, agents can only move to adjacent (or the same) vertices, and must avoid either vertex or swap conflicts, while maintaining the connectivity of the entire configuration.
Formally, a \emph{valid} plan $\Pi$ satisfies the following conditions:
\begin{enumerate}
   \item $\conf{0} = S\land \conf{t^{\ast}} = T$.
   \item $\forall i \in A, \forall k \in [1, t^{\ast}]: \mathcal{Q}_{k}[i] \in \nei{\mathcal{Q}_{k-1}[i]}\cup \{\mathcal{Q}_{k-1}[i]\}$. (\emph{reachability})
   \item $\forall k \in [0,t^{\ast}], \forall i, j \in A, i \neq j: \mathcal{Q}_k[i] \neq \mathcal{Q}_k[j]$. (avoidance of \emph{vertex conflict})
   \item $\forall k \in [1,t^{\ast}], \forall i, j \in A, i \neq j: (\conf{k}[i],\conf{k}[j])\ne (\conf{k-1}[j],\conf{k-1}[i])$.  (avoidance of \emph{swap conflict})
   \item $\forall k \in [0,t^{\ast}]: \conf{k}$ is connected.
\end{enumerate}
We aim to find a plan $\Pi$ that minimizes the makespan $t^{\ast}$.

\section{Related Work}
\paragraph{Unlabeled MAPF (a.k.a. \emph{Anonymous MAPF}).}
Unlike the labeled setting, the unlabeled variant introduces the additional subproblem of assigning agents to target locations.
Since the optimal goal assignment depends on the resulting path, these two subproblems---assignment and pathfinding---are tightly coupled and solved simultaneously.
The seminal work~\cite{Yu13unlabeled} showed that optimal solutions can be found in polynomial time via a reduction to the \textsc{Maximum Flow} problem.
However, this flow-based approach often suffers from high computational costs on large instances in practice, limiting its scalability.
These scalability challenges have motivated the development of fast suboptimal algorithms~\cite{TSWAP23}.

\paragraph{MAPF with additional constraint.}
\mapf typically considers vertex or swap conflicts; however, practical applications often require handling additional constraints due to communication or other limitations; for instance, \citet{Li20} aims for agents to move while maintaining a specific formation, and \citet{Kristan25} conducts theoretical studies on combinatorial constraints on graphs, e.g., dominating set and independent set.
\citet{Tateo2018} show the PSPACE-hardness in the case where the physical and communication features are given separately.

\paragraph{\prbname.}
\citet{Fekete21} show that finding optimal solutions to \prbname is NP-hard, even in 2D empty grid graphs.
This hardness also applies to finding a good solution, specifically a 3/2-approximation of the makespan.
The work also introduces a polynomial-time, constant-factor approximation algorithm for makespan in \emph{empty} 2D grid environments.
Despite this, its practical applicability is severely limited, as it requires a large number of theoretical steps, has quite a large approximation ratio, and works only under restricted conditions, motivating our study.

\section{Optimal Approach: ILP}
While the makespan-optimal unlabeled \mapf can be solved by reduction to the \textsc{Maximum Flow} problem~\cite{Yu13unlabeled}, the same idea does not apply to \prbname due to the connectivity constraint.
We thus adopt the reduction to \textsc{Integer Linear Programming} (\ilp).
This section first formulates \ilp for the unlabeled case, and extends it to include the connectivity constraint.

\subsection{ILP Formulation of Unlabeled \mapf}

In this section, we describe a reduction to the unlabeled \mapf to \ilp.
Note that unlabeled MAPF can be reduced to \textsc{Maximum Flow}, which in turn can be reduced to an \ilp; however, for simplicity, we directly give an \ilp formulation.

First, we define the \emph{Bounded} variant of unlabeled \mapf: given a graph $G=(V,E)$, vertex subset $S,T \subseteq V$, and \emph{time bound} $\tau\in \mathbb{N}$, it determines whether there exists a plan with conditions 1--4 in \Cref{sec: preliminaries} and the length is at most $\tau$.

Here, we assume that all variables are nonnegative.
Let $x_{v,t}$ be a variable that indicates whether an agent is at vertex $v$ at step $t$, that is, $x_{v,t}=1$ when $v\in \conf{t}$, and 0 otherwise.
Also, $f_{u,v,t}$ represents whether an agent at vertex $u$ at step $t$ moves to $v\in \neicl{u}$ at timestep $t+1$, using an edge $uv\in E$.
Then, unlabeled \mapf can be formulated as follows.
\begin{enumerate}
    \item For every $u,v\in V$, $x_{v,t} \le 1$.
    \item Initial configuration $\conf{0}$ is $S$: for $v\in V$, $x_{v,0}=1$ if $v\in S$, and 0 otherwise.
    \item Final configuration $\conf{\tau}$ is $T$: for $v\in V$, $x_{v,\tau}=1$ if $v\in T$, and 0 otherwise.
    \item For every timestep $t\in [0,\tau]$, at most $x_u$ agents on $u$ move to some vertex $v\in \neicl{u}$: for every $u\in V$, $\sum_{v\in N[u]} f_{u,v,t} = x_{u,t}$
    \item For every timestep $t\in [1,\tau+1]$, at most $x_u$ agents on $u$ come from some vertex $v\in \neicl{u}$: for every $u\in V$, $\sum_{v\in N[u]} f_{v,u,t-1} = x_{u,t}$
\end{enumerate}
Then we set an objective as a constant, to check whether these constraints are \emph{feasible}, i.e., there is a satisfying assignment for these constraints.
Note that swap conflicts can be resolved by a common postprocessing~\cite{Yu13unlabeled}.

By solving this bounded unlabeled \mapf incrementally from $\tau=1$ at most $|V|+n-2$ times, there exists at least one feasible instance~\cite{Yu13unlabeled}.
Let $\tau'$ be the minimum time bound in the instance of unlabeled \mapf. 
Note that the assignment of $x_{v,t}$ corresponds to the position of each agent at each timestep.
Based on the assignment for this instance, we can construct a valid plan $\Pi$ with length at most $\tau'$ by a proper postprocessing.%
\footnote{Since the information about agent indices has been lost, we use a perfect matching to recover it. Detailed implementations are in \url{https://github.com/su3taka/CUMAPF-PULL}.}
Thus, this formulation leads to an optimal algorithm.

\subsection{Connectivity Constraint}
One can observe that all we have to do is to add the constraint to $V_t\coloneqq \{v\mid x_{v,t}=1\}$ for every $t\in [0,\tau]$ such that $G[V_t]$ is connected, i.e., there is at least one $(u, v)$-path for every $u,v\in V_t$.
The connectivity of the graph is encoded by a \emph{Single-commodity Flow} model \cite{Conrad07}: a single source at the root supplies one commodity to each selected vertex, with flow allowed only through the vertices in $V_t$.

First, we introduce variables $r_{v,t}\in [0,x_{v,t}]$ that represent whether $v$ is chosen as a source, and $\ell_{uv,t}$ indicating the number of commodities that are carried through an edge from $u$ to $v$ such that $uv\in E$.
Then, we can formulate the constraints as follows;
\begin{enumerate}
    \item $\forall v\in V$, $r_{v,t}\le x_{v,t}$,
    \item $\sum_{v\in V}r_{v,t}=1$,
    \item $\forall u,v\in V$. $\ell_{uv,t}\le (|A|-1)x_{v,t}$, and
    \item $\forall v\in V$. $\sum_{u\in N[v]}(\ell_{uv,t}-\ell_{vu,t})=x_{v,t}- |A|r_{v,t}$.
\end{enumerate}
We now justify these constraints, with reference to Fig.~\ref{fig:ilpidea}.
First, in Constraints 1 and 2, we can restrict the source of flow to a single vertex per timestep.
The last two constraints coincide with carrying commodities via a flow network that consists of vertices in $V_t$.
To do so, Constraint 3 restricts usable edges to those that connect only vertices in $V_t$; if $x_{u,t}=x_{v,t}=1$, then edge $uv$ can carry 
$|A|-1$ units of flow, that is, all the flow except the $(u,u)$-path.
Regarding the last Constraint, consider adding unit of flow that goes along some $(v_1,v_2)$-path to the current flow, and see the value $\sum_{u\in N[v]}(\ell_{uv,t}-\ell_{vu,t})$ of the left-hand side in this case. 
At the source $v_1$, one unit of flow leaves, so the value decreases by 1; at $v_2$, one unit of flow enters, so the value increases by 1; and at intermediate vertices, one unit enters, and one unit leaves, thus the net change is 0.
Here, in our network, we start from an empty flow, and add a $(u,v)$-path for the root $u$ with $r_{u,t}=1$ and every $v \in V_t$. Therefore, the final value is $1-|A|$ at the root $u$, 1 for every $v\in V_t$, and 0 for every $v\notin V_t$.
This is because $u$ is selected $|A|$ times as an initial vertex and once as a target vertex, and $v\in V_t$ is selected once as a target vertex.
The value $x_{v,t}- |A|r_{v,t}$ of the right-hand side coincides exactly with this value.
The optimality is justified in the same way as for unlabeled \mapf.

\begin{figure}[ht]
    \centering
\begin{tikzpicture}[>=stealth,scale=1]
\tikzset{root/.style={circle,fill,inner sep=2pt}}
\tikzset{inner/.style={circle, draw,fill=white,inner sep=2pt}}
\tikzset{leaf/.style={circle, draw,fill=white,inner sep=2pt}}
\tikzset{edge/.style={very thick}}
\tikzset{tedge/.style={semithick}}
\tikzset{redge/.style={very thick, red}}
\begin{scope}[shift={(-5,0)}]
  \node[root] (L) at (0,0) {};

  \node[inner] (v1) at (0.8,1) {};
  \node[inner] (v2) at (1,0) {};
  \node[inner] (v3) at (0.7,-0.5) {};
  \node[inner] (v4) at (1.5,0.6) {};
  \node[inner] (v5) at (1.8,-0.7) {};
  \node[inner] (v6) at (2,1) {};
  \node[inner] (v7) at (2,-0.2) {};
  \node[inner] (v8) at (2.5,0.2) {};
  \draw[edge] (L) -- (v1);
  \draw[edge] (L) -- (v2);
  \draw[edge] (L) -- (v3);
  \draw[edge] (v1) -- (v4);
  \draw[edge] (v4) -- (v6);
  \draw[edge] (v2) -- (v7);
  \draw[edge] (v3) -- (v5);
  \draw[tedge] (v5) -- (v7);
  \draw[tedge] (v4) -- (v7);
  \draw[tedge] (v8) -- (v7);
  \draw[tedge] (v2) -- (v4);
  \draw[tedge] (v2) -- (v3);
  \draw[tedge] (v8) -- (v7);
  \draw[tedge] (v8) -- (v6);
\end{scope}

\node at (-1.5,0) {$\Rightarrow$};

\begin{scope}[shift={(-0.5,0)}]
  \node[root] (L) at (0,0) {};

  \node[inner] (v1) at (0.8,1) {};
  \node[inner] (v2) at (1,0) {};
  \node[inner] (v3) at (0.7,-0.5) {};
  \node[inner] (v4) at (1.5,0.6) {};
  \node[inner] (v5) at (1.8,-0.7) {};
  \node[inner] (v6) at (2,1) {};
  \node[inner] (v7) at (2,-0.2) {};
  \node[inner] (v8) at (2.5,0.2) {};
  \draw[edge] (L) -- (v1);
  \draw[edge] (L) -- (v2);
  \draw[edge] (L) -- (v3);
  \draw[edge] (v1) -- (v4);
  \draw[edge] (v4) -- (v6);
  \draw[edge] (v2) -- (v7);
  \draw[edge] (v3) -- (v5);
  \draw[tedge] (v5) -- (v7);
  \draw[tedge] (v4) -- (v7);
  \draw[tedge] (v8) -- (v7);
  \draw[tedge] (v2) -- (v4);
  \draw[tedge] (v2) -- (v3);
  \draw[tedge] (v8) -- (v7);
  \draw[tedge] (v8) -- (v6);
  \draw[redge] (L) -- (v3);
  \draw[redge] (v3) -- (v2);
  \draw[redge] (v2) -- (v7);
  \draw[redge] (v7) -- (v8);
  
  \node at (-0.4,0) {{\scriptsize $-1$}};
  \node at (0.4,-0.7) {{\scriptsize $\pm 0$}};
  \node at (0.8,0.2) {{\scriptsize $\pm 0$}};
  \node at (2.3,-0.3) {{\scriptsize $\pm 0$}};
  \node at (2.6,0.5) {{\scriptsize $+1$}};
\end{scope}

\end{tikzpicture}
    \caption{Constraint 4 for \ilp reduction. Circles represent vertices of the graph, and a black circle represents the sources determined by Constraint 1. Lines represent edges, and bold lines indicate edges used by the flow. (\emph{Left}) An intermediate network and several flows. (\emph{Right}) The network with a new flow (the red path). We indicate the net change of $\sum_{u\in N[v]}(\ell_{uv,t}-\ell_{vu,t})$ for each vertex on the red path.}
    \label{fig:ilpidea}
\end{figure}

\subsection{Limitation – Scalability}\label{sec:expilp}
One can observe that, 
given an instance of bounded \prbname, we solve an \ilp with $O((|V|+|E|)\tau)$ variables and $O(|V|\tau)$ constraints.

\begin{table}[t]
    \centering
\caption{Evaluation on the optimal algorithm, using 300 \SI{}{\second} timeout. For each map, we prepared five types of $n$, generated 50 random instances for each. The time column reports the average time to find a solution over the solved instances. We also report the average running time of \algname (\SI{}{\second}), the ratio of time improvement (\ilp/\algname), and the suboptimality of makespan (\algname/\ilp) on the instances successfully solved by \ilp.
The ratio is omitted in the last row, due to the small sample size.}\label{tab:ILP}
\setlength{\tabcolsep}{2pt}
\renewcommand{\arraystretch}{1}
\begin{adjustbox}{max width=\textwidth}
{\small
\begin{tabular}{crrrrrr}
\toprule
&& \multicolumn{1}{c}{\ilp} &\multicolumn{2}{c}{time(\SI{}{\second})}&\multicolumn{2}{c}{ratio}\\
\cmidrule{4-7}
  Map & $n$ & solved (\%) & \ilp & \algname & speedup& makespan\\
\midrule
  \makecell[c]{{\scriptsize\emph{empty-8-8}}\\ \adjustbox{raise=-5mm}{\includegraphics[width = 0.15\linewidth]{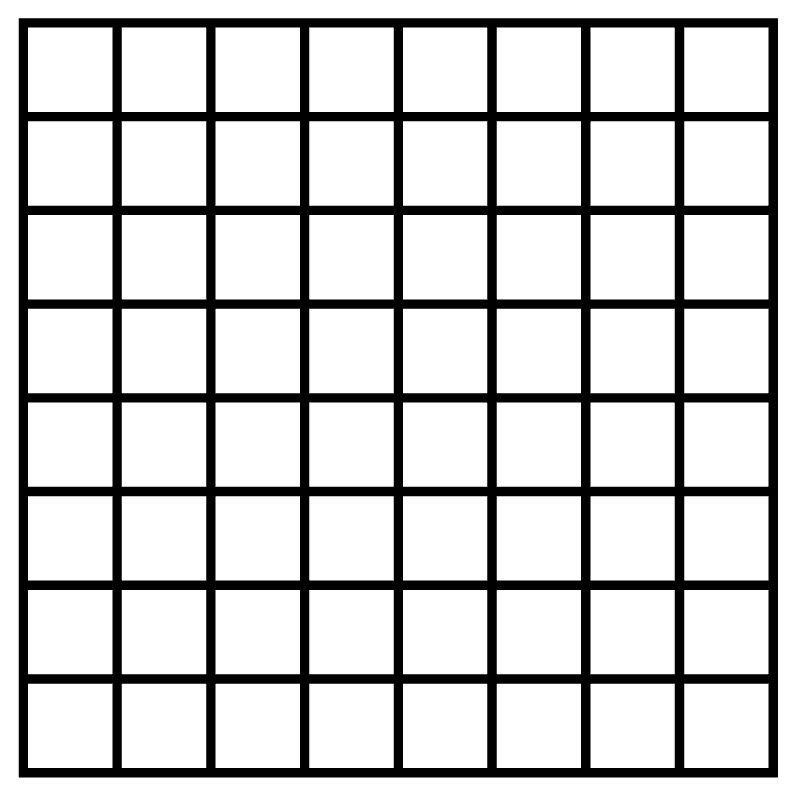}}} 
& \makecell[r]{10\\15\\20\\25\\30}&\makecell[r]{100\\100\\100\\100\\100}&\makecell[r]{0.137\\0.226\\0.233\\0.202\\0.190}  &\makecell[r]{0.005\\0.010\\0.013\\0.015\\0.018} &
\makecell[r]{27.36\\18.73\\15.15\\11.95\\9.934}&
\makecell[r]{1.347\\1.484\\1.540\\1.678\\1.713}\\
\midrule
  \makecell[c]{{\scriptsize\emph{random-32-32-20}}\\\adjustbox{raise=-5mm}{\includegraphics[width = 0.15\linewidth]{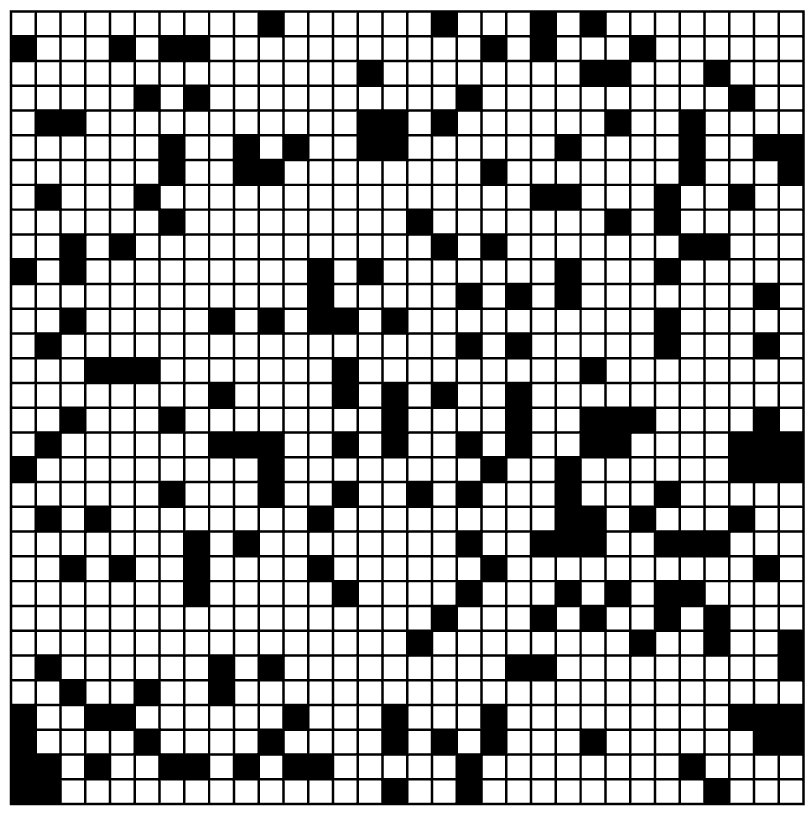}}} &
\makecell[r]{10\\20\\30\\40\\50}&\makecell[r]{82\\48\\34\\24\\4}  &
\makecell[r]{56.34\\75.36\\57.43\\77.75\\93.66} &\makecell[r]{0.020\\0.050\\0.094\\0.157\\0.231} &
\makecell[r]{1660\\1499\\862.1\\804.1\\N/A}&
\makecell[r]{1.131\\1.460\\1.650\\1.676\\N/A}\\
\bottomrule
\end{tabular}
}
\end{adjustbox}
\end{table}

Although this method can efficiently find optimal solutions on small instances, the computational costs on large-scale problems become prohibitively high.
\Cref{tab:ILP} shows the empirical results for two maps retrieved from the well-known MAPF benchmarks~\cite{Stern19}, using Gurobi as the ILP solver.
On a small map (\emph{empty-8-8}), the solver can find an optimal solution within a second, while on a larger map (\emph{random-32-32-20}), the solver often fails to find a solution, and the number of failed instances is increasing as $n$ increases.
Moreover, the average time to find a solution also increases as $n$ increases.
This empirical result motivates us to explore a fast and scalable suboptimal approach, which we will introduce next.

\section{Suboptimal Approach: \algname}

This section proposes \algname, a suboptimal \prbname algorithm. The algorithm is designed as a configuration generator that takes the current agent configuration $\qfrom$ and the target vertex set $T$ as input to compute the configuration for the next timestep. By iteratively applying this function from $S$, the algorithm is guaranteed to converge to $T$ in a finite number of steps.
The pseudocode for the algorithm is provided in \Cref{alg:PULL}. The following subsections supplement the pseudocode by explaining its main structure and core mechanisms.
Herein, $\textsf{next}(a,b)$ denotes the parent vertex of $a\in V$ on the BFS-tree rooted by $b\in V$.
Cleary, we have $\dist(\textsf{next}(a,b),b)=\dist(a,b)-1$ when $a\ne b$.

\subsection{Concept}\label{sec:concept}

\begin{figure}
    \centering
    \begin{tikzpicture}[
    xscale = 0.65,
    yscale = 0.55,
    transform shape,
    base_ellipse/.style={ellipse, draw, minimum width=3cm, minimum height=0.7cm},
    red_ellipse/.style={base_ellipse, draw=red, fill=red!20},
    blue_ellipse/.style={base_ellipse, draw=blue, fill=blue!20},
    red_dot/.style={circle, fill=red, draw=red, inner sep=2pt},
    black_dot/.style={circle, fill=black, draw=black, inner sep=2pt}
]

\begin{scope}[xshift=0cm]
    \node[blue_ellipse] (blue1) at (0,0.7) {};
    \node[red_ellipse, fill=red!60] (red1) at (0,4) {};
    \draw[->, thick] (red1.south) -- (blue1.north);
    \node[fill=white, align=center, rotate = 90, font = \large]
    at (0, 2.3) {$\cdots$};
    \node[fill=white, align=center, font = \huge]
    at (-1.5, 4.73) {$\conf{0}$};
    \node[fill=white, align=center, font = \huge]
    at (-1.1, 1.34) {$T$};
    \node[black_dot] (x1) at (0, 3.7) {};
    \node[align=center, font = \huge]
    at (0.5, 3.3) {$x$};
    \node[black_dot] (y1) at (0, 0.9) {};
    \node[align=center, font = \huge]
    at (0.5, 1.3) {$y$};
\end{scope}

\begin{scope}[xshift=4cm]
    \node[blue_ellipse] (blue2) at (0,0.7) {};
    \node[red_ellipse] (red2) at (0,4) {};
    \node[base_ellipse, draw=red, fill=red!60, minimum width=2cm, minimum height=0.5cm, yshift = -1mm] at (red2.center) {};
    \draw[->, thick] (red2.south) -- (blue2.north);
    \node[red_dot] at (0, 3.5) {};
    \node[fill=white, align=center, rotate = 90, font = \large]
    at (0, 2.3) {$\cdots$};
    \node[fill=white, align=center, font = \huge]
    at (-1, 4.73) {$\conf{1}$};
    \node[fill=white, align=center, font = \huge]
    at (-1.1, 1.34) {$T$};
    \node[black_dot] (x1) at (0, 3.7) {};
    \node[black_dot] (y1) at (0, 0.9) {};
\end{scope}

\begin{scope}[xshift=8cm]
    \node[blue_ellipse] (blue3) at (0,0.7) {};
    \node[red_ellipse] (red3) at (0,4) {};
    \draw[->, thick] (red3.south) -- (blue3.north);
    \foreach \y in {3.5,3.2, 2.9, 1.7, 1.4, 1.1} {
        \node[red_dot] at (0, \y) {};
    }
    \node[fill=white, align=center, rotate = 90, font = \large]
    at (0, 2.3) {$\cdots$};
    \node[fill=white, align=center, font = \huge]
    at (-1, 4.73) {$\conf{t'}$};
    \node[fill=white, align=center, font = \huge]
    at (-1.1, 1.34) {$T$};
    \node[black_dot] (x1) at (0, 3.7) {};
    \node[black_dot] (y1) at (0, 0.9) {};
\end{scope}

\draw[->, gray, very thick] (1, 2) -- (2.5, 2);
\draw[->, dashed, gray, very thick] (5.5, 2) -- (7, 2);

\end{tikzpicture}
    \caption{The naive approach for solving \prbname. Light red represents $\qfrom=\setconf{0}$, deep red represents the current configuration, and blue represents $T$. (\emph{Left}) Initial configuration, 
    (\emph{Center}) generated configuration $\conf{1}$ from $\conf{0}$, (\emph{Right}) configuration at step $t'=\dist(x,y)$.}
    \label{fig:idea}
\end{figure}

We begin the description of \algname with a bare minimum idea to solve \prbname.
Here, given an input configuration $\qfrom$ that is not duplicated with a target $T$ (i.e., $\qfrom \cap T = \emptyset$), we are interested in obtaining a transitionable configuration $\current$ from $\qfrom$ that is closer to $T$.

Let $x\in \qfrom$ and $y\in T$ be a pair of vertices such that $(x,y)= \argmin_{(s,t)\in \qfrom\times T} \dist(s,t)$ (\Cref{fig:idea} left).
If the agent currently on $x$ can move to $z\coloneqq \textsf{next}(x,y)$ without breaking connectivity, then we can guarantee that there is at least one agent that moves towards the goal.

If $(\qfrom\setminus\{x\})\cup \{z\}$ is not connected, then we have to move another agent on $u\in \nei{x}$ to $x$. 
This operation might generate another disconnected configuration; then we can repeat the same process until the configuration is connected. 
This recursive procedure allows the moved agent to form a single path that follows the DFS-tree over connected agents.
It always terminates when it reaches a vertex that does not ``cut'' the configuration, since any leaf of a DFS-tree cannot be a cut vertex.
We can ensure that the agent closest to its goal in $\qfrom$ moves toward $T$, and we obtain a configuration that is ``closer'' to the goals than $\qfrom$ (\Cref{fig:idea} center).

By iteratively applying this process, the agent moves a distance of one per step on the shortest $(x,y)$-path towards $y$.
Then, there is a step $t$ such that exactly one agent reaches $y$ (\Cref{fig:idea} right).
The remaining procedure is to ``pour'' the agents into the goal, which will be elaborated in \Cref{sec:alg}.

While this naive approach can solve CUMAPF, it provides a highly makespan-suboptimal solution in practice.
This is because the algorithm only moves the agents along a simple path, starting from $\textsf{next}(x,y)$, at each step, and other agents remain unchanged.
For example, it only executes the leftmost column in \Cref{fig:execution}, resulting in the configuration shown in the upper part of the second column.

The aforementioned idea uses a single path in the induced subgraph $\indG{\setqfrom}$ when generating the next configuration.
However, there may exist multiple paths towards the goals as long as they are pairwise disjoint and the resulting configuration $\current$ is connected.
That is, we can synthesise a configuration much closer to the goal by moving other unconstrained agents.
This idea constitutes the approach of \algname.

\subsection{Algorithm}\label{sec:alg}
{\small
	\begin{algorithm}[t]
		\caption{\algname}
		\label{alg:PULL}
		\begin{algorithmic}[1]  
			\Require configuration $\qfrom$, agents $A$, goals $T$
			\Ensure configuration $\current$ \Comment{initialized with $\qfrom$}
            \State $R\leftarrow \emptyset$\label{pull:init}

            \medskip
            \Statex {\small\emph{(when the configuration and goals overlapped)}}
            \If{$\setqfrom\cap T \ne \emptyset$} \label{pull:goalornot}
            \State $\mathcal{P}\leftarrow$ Components of $(G[\setqfrom\cap T])$ \label{pull:goal_s}
            \State Sort $P_i\in \mathcal{P}$ in descending order of $|P_i|$\label{pull:sortcc}
            \For{$P_i\in \mathcal{P}$}\label{zhoge1}
            \For{$v\in (N(P_i) \cap T)$} \label{zhoge2}
            \If{$\nexists j$ s.t. $\current[j]=v$}\label{pull: goal_call} $R \leftarrow \procname(v, R, P_i)$ \label{zhoge3}
            \EndIf
            \EndFor
            \For{$v\in P_i$}
            \If{$\exists i$ s.t. $\current[i]=\qfrom[i]$} $R\leftarrow R\cup \{i\}$\label{pull:goal_e}
            \EndIf
            \EndFor
            \EndFor 
            \EndIf\label{pull:goalornot-end}
            \medskip
            \Statex {\small\emph{(main operation)}}
            \State Sort $N(\setqfrom)$ by $\min_{t\in T}(\dist(u,t))$ in ascending order\label{pull:notg_s}
            \For{$u\in N(\setqfrom)$}\label{zhoge4}
            \If{$\nexists j$ s.t. $\current[j]=u$} $R\leftarrow \procname(u,R, \emptyset)$ \label{pull:notg_e}
            \EndIf
            \EndFor
            \State\Return $\current$\label{zhoge5}

            \medskip
            \Statex {\small\emph{(pull agents from $t \in V$)}}
            \Procedure{\procname}{$t,R, V'$}\label{pull:proc-start}
            \State $\reachable\leftarrow \bfs{G[\setcurrent\cup \{t\}\setminus \setcurrent(\reserved)]}{t}$\label{pull:BFS}
            
            \State $\cutset\leftarrow \textsf{Cut}(G[\setcurrent\cup \{t\}])$ \label{pull:DFS}
            \If{$V_S\leftarrow (\reachable\setminus \cutset) \setminus (V'\cup \{t\})= \emptyset$} \Return $R$ \label{pull:judge}
            \EndIf
            \State $\textsf{cur} \leftarrow \argmax_{v \in V_S} ( \min_{y\in T} (\dist(v,y))$\label{pull:select}
            \While{$\textsf{cur} \ne t$}
            \State $i\leftarrow$ agents s.t. $\qfrom[i]=\textsf{cur}$
            \State $Q^{to}[i]\leftarrow \textsf{next}(\textsf{cur},t)$\Comment{determining $\current$}\label{pull:assign}
            \State $R\leftarrow R\cup \{i\}$, $\textsf{cur} \leftarrow Q^{to}[i]$
            \EndWhile 
            \State \Return $R$
            \EndProcedure\label{pull:proc-end}
		\end{algorithmic}
	\end{algorithm}
}

\Cref{alg:PULL} illustrates \algname, which consists of three blocks:
\emph{(i)}~Lines~\ref{pull:proc-start}--\ref{pull:proc-end} denote a subprocedure \procname that ``pulls'' the agents,
\emph{(ii)}~Lines~\ref{pull:notg_s}--\ref{zhoge5} call \procname to repeatedly identify pullable paths, and;
\emph{(iii)}~Lines~\ref{pull:goalornot}--\ref{pull:goalornot-end} handle a configuration that overlaps with the goal configuration by calling \procname with the non-movable, constrained agents.

Our core is the procedure $\procname(t, R, V')$ for $t\in \nei{\qfrom}$, $R\subseteq A$, $V'\subseteq \setqfrom$. 
This computes a path on $\indG{\setqfrom}$ to ``pull'' agents toward a specified vertex $t$, while avoiding the subset of agents $\reserved$ and preserving connectivity.
Furthermore, the path never starts from any vertex in $V'$.
\Cref{fig:execution} provides an example of the whole process to output $\current$ from $\qfrom$.
First, we provide detailed descriptions of \procname.

\begin{figure}[t]
    \centering

\begin{tikzpicture}[scale = 0.67]
    \def\cellwidth{0.6}
    \def\cellheight{0.6}

    \newcommand{\drawgrid}[2]{
        \foreach \i in {0,...,5} {
            \draw (#1,\i * -\cellheight + #2) -- (3 * \cellwidth + #1,\i * -\cellheight + #2);
        }
        \foreach \j in {0,...,3} {
            \draw (\j * \cellwidth + #1,#2) -- (\j * \cellwidth + #1,#2 - 5 * \cellheight);
        }
    }
    \newcommand{\drawnumbers}[2]{
        \foreach \c in {1,...,3} {
            \pgfmathsetmacro{\x}{#2 + (\c - 0.5)*\cellwidth}
            \pgfmathsetmacro{\y}{#1 + 0.3} 
            \node[font=\small] at (\x, \y) {\c};
        }
        \foreach \r in {1,...,5} {
            \pgfmathsetmacro{\x}{#2 - 0.3} 
            \pgfmathsetmacro{\y}{#1 - (\r - 0.5)*\cellheight}
            \node[font=\small] at (\x, \y) {\r};
        }
    }

    \newcommand{\drawredcircles}[3]{
        \foreach \r/\c in #3 {
            \pgfmathsetmacro{\x}{#2 + (\c - 0.5)*\cellwidth}
            \pgfmathsetmacro{\y}{#1 - (\r - 0.5)*\cellheight}
            \fill[red] (\x,\y) circle (5pt);
        }
    }

    \newcommand{\drawlightredcircles}[3]{
        \foreach \r/\c in #3 {
            \pgfmathsetmacro{\x}{#2 + (\c - 0.5)*\cellwidth}
            \pgfmathsetmacro{\y}{#1 - (\r - 0.5)*\cellheight}
            \fill[red!50] (\x,\y) circle (5pt);
        }
    }

    \newcommand{\drawbluecells}[2]{
        \foreach \r/\c in {5/1,4/1,5/2,4/2,3/2,5/3,4/3} {
            \pgfmathsetmacro{\x}{#2 + (\c - 1)*\cellwidth}
        \pgfmathsetmacro{\y}{#1 - (\r - 1)*\cellheight}
        \fill[blue!10] (\x, \y) rectangle ++(\cellwidth, -\cellheight);
        }
    }
    \newcommand{\drawyellowcells}[3]{
        \foreach \r/\c in #3 {
            \pgfmathsetmacro{\x}{#2 + (\c - 1)*\cellwidth}
        \pgfmathsetmacro{\y}{#1 - (\r - 1)*\cellheight}
        \fill[yellow] (\x, \y) rectangle ++(\cellwidth, -\cellheight);
        }
    }

    \newcommand{\drawdashedcircles}[3]{
        \foreach \r/\c in #3 {
            \pgfmathsetmacro{\x}{#2 + (\c - 0.5)*\cellwidth}
            \pgfmathsetmacro{\y}{#1 - (\r - 0.5)*\cellheight}
            \draw[draw=black, fill=white, dashed, thick] (\x,\y) circle (5pt);
        }
    }
     \newcommand{\drawdashedyellowcircles}[3]{
        \foreach \r/\c in #3 {
            \pgfmathsetmacro{\x}{#2 + (\c - 0.5)*\cellwidth}
            \pgfmathsetmacro{\y}{#1 - (\r - 0.5)*\cellheight}
            \draw[draw=black, fill=yellow, dashed, thick] (\x,\y) circle (5pt);
        }
    }

    \newcommand{\drawvarrows}[3]{
        \foreach \rfrom/\cfrom/\rto/\cto in #3 {
            \pgfmathsetmacro{\xfrom}{#2 + (\cfrom - 0.5)*\cellwidth}
            \pgfmathsetmacro{\yfrom}{#1 - (\rfrom - 0.3)*\cellheight}
            \pgfmathsetmacro{\xto}{#2 + (\cto - 0.5)*\cellwidth}
            \pgfmathsetmacro{\yto}{#1 - (\rto - 0.7)*\cellheight}
            \draw[->, thick] (\xfrom,\yfrom) -- (\xto,\yto);
        }
    }

    \newcommand{\drawharrows}[3]{
        \foreach \rfrom/\cfrom/\rto/\cto in #3 {
            \pgfmathsetmacro{\y}{#1 - (\rfrom - 0.5)*\cellheight}
            \pgfmathsetmacro{\xfrom}{#2 + (\cfrom - 0.3)*\cellwidth}
            \pgfmathsetmacro{\xto}{#2 + (\cto - 0.7)*\cellwidth}
            
            \draw[->, thick] (\xfrom,\y) -- (\xto,\y);
        }
    }

    \def\xA{0}  \def\yA{0}
    \drawbluecells{\yA}{\xA}
    \drawgrid{\xA}{\yA}
    \drawnumbers{\yA}{\xA}
    \drawlightredcircles{\yA}{\xA}{{4/1,1/2,4/2,2/1,2/2,2/3,3/2}}
    \drawdashedyellowcircles{\yA}{\xA}{{4/3}}

    \def\xB{3.2}  \def\yB{0}
    \drawbluecells{\yB}{\xB}
    \drawgrid{\xB}{\yB}
    \drawnumbers{\yB}{\xB}
    \drawlightredcircles{\yB}{\xB}{{4/1,2/1,2/3}}
    \drawredcircles{\yB}{\xB}{{2/2,3/2,4/2,4/3}}
    \drawdashedcircles{\yB}{\xB}{{1/2}}
    \drawdashedyellowcircles{\yB}{\xB}{{5/1}}
    \drawvarrows{\yB}{\xB}{{1/2/2/2, 2/2/3/2, 3/2/4/2}}
    \drawharrows{\yB}{\xB}{{4/2/4/3}}

    \def\xC{6.4}  \def\yC{0}
    \drawbluecells{\yC}{\xC}
    \drawgrid{\xC}{\yC}
    \drawnumbers{\yC}{\xC}
    \drawlightredcircles{\yC}{\xC}{{2/1,2/3}}
    \drawredcircles{\yC}{\xC}{{4/1,2/2,3/2,4/2,4/3}}
    \drawdashedcircles{\yC}{\xC}{{1/2}}
    \drawdashedyellowcircles{\yC}{\xC}{{3/1}}
    \drawvarrows{\yC}{\xC}{{1/2/2/2, 2/2/3/2, 3/2/4/2}}
    \drawharrows{\yC}{\xC}{{4/2/4/3}}

    \def\xD{0}  \def\yD{-4.8}
    \drawbluecells{\yD}{\xD}
    \drawyellowcells{\yD}{\xD}{{1/2}}
    \drawgrid{\xD}{\yD}
    \drawnumbers{\yD}{\xD}
    \drawlightredcircles{\yD}{\xD}{{4/1,2/1,2/3,1/2,4/2,2/2,3/2}}
    \drawdashedcircles{\yD}{\xD}{{4/3}}
    \drawvarrows{\yD}{\xD}{{1/2/2/2, 2/2/3/2, 3/2/4/2}}
    \drawharrows{\yD}{\xD}{{4/2/4/3}}

    \def\xE{3.2}  \def\yE{-4.8}
    \drawbluecells{\yE}{\xE}
    \drawgrid{\xE}{\yE}
    \drawnumbers{\yE}{\xE}
    \drawlightredcircles{\yE}{\xE}{{2/1,2/3}}
    \drawredcircles{\yE}{\xE}{{2/2,3/2,4/2,4/3,4/1}}
    \drawdashedcircles{\yE}{\xE}{{1/2}}
    \drawvarrows{\yE}{\xE}{{1/2/2/2, 2/2/3/2, 3/2/4/2}}
    \drawharrows{\yE}{\xE}{{4/2/4/3}}

    \def\xF{6.4}  \def\yF{-4.8}
    \drawyellowcells{\yF}{\xF}{{2/1}}
    \drawbluecells{\yF}{\xF}
    \drawgrid{\xF}{\yF}
    \drawnumbers{\yF}{\xF}
    \drawlightredcircles{\yF}{\xF}{{2/1,2/3}}
    \drawredcircles{\yF}{\xF}{{4/1,2/2,3/2,4/2,4/3}}
    \drawdashedcircles{\yF}{\xF}{{1/2,3/1}}
    \drawvarrows{\yF}{\xF}{{1/2/2/2, 2/2/3/2, 3/2/4/2, 2/1/3/1}}
    \drawharrows{\yF}{\xF}{{4/2/4/3}}
    
    \def\xG{9.6}  \def\yG{0}
    \drawbluecells{\yG}{\xG}
    \drawgrid{\xG}{\yG}
    \drawnumbers{\yG}{\xG}
    \drawlightredcircles{\yG}{\xG}{{2/3}}
    \drawredcircles{\yG}{\xG}{{3/1,4/1,2/2,3/2,4/2,4/3}}
    \drawdashedcircles{\yG}{\xG}{{1/2,2/1}}
    \drawvarrows{\yG}{\xG}{{1/2/2/2, 2/2/3/2, 3/2/4/2, 2/1/3/1}}
    \drawharrows{\yG}{\xG}{{4/2/4/3}}

    \def\xH{9.6}  \def\yH{-4.8}
    \drawbluecells{\yH}{\xH}
    \drawgrid{\xH}{\yH}
    \drawnumbers{\yH}{\xH}
    \drawredcircles{\yH}{\xH}{{3/1,4/1,2/2,3/2,4/2,4/3,3/3}}
    \drawdashedcircles{\yH}{\xH}{{1/2,2/1,2/3}}
    \drawvarrows{\yH}{\xH}{{1/2/2/2, 2/2/3/2, 3/2/4/2, 2/1/3/1,2/3/3/3}}
    \drawharrows{\yH}{\xH}{{4/2/4/3}}

    \node at (0.9, 0.8) {$t = v_{\texttt{4-3}}$};
    \node at (4.1, 0.8) {$t = v_{\texttt{5-1}}$};
    \node at (7.3, 0.8) {$t = v_{\texttt{3-1}}$};

    \node at (0.9, -8.15) {$\mathsf{cur} = v_{\texttt{1-2}}$};
    \node at (4.1, -8.15) {$v_{\texttt{4-1}}$: stay};
    \node at (7.3, -8.15) {$\mathsf{cur} = v_{\texttt{2-1}}$};
    \node at (10.5, -8.15) {output: $\current$};

    \node at (0.9, -3.4) {$V_S\ne \emptyset$};
    \node at (4.1, -3.4) {$V_S=\emptyset$};
    \node at (7.3, -3.4) {$V_S\ne \emptyset$};

    \draw[->, very thick] (0.9, -3.7) -- (0.9, -4.3);
    \draw[->>, very thick] (4.1, -3.7) -- (4.1, -4.3);
    \draw[->, very thick] (7.3, -3.7) -- (7.3, -4.3);
    \draw[->>, very thick] (10.5, -3.2) -- (10.5, -4.3);

    \draw[->, very thick]
        (2, -7) .. controls (3, -7) and (1.7, -2) .. (2.7, -2);
    \draw[->, very thick]
        (5.2, -7) .. controls (6.2, -7) and (4.9, -2) .. (5.9, -2);
    \draw[->, very thick]
        (8.4, -7) .. controls (9.4, -7) and (8.1, -2) .. (9.1, -2);
\end{tikzpicture}

    \caption{An example of the \algname's operation. The graph $G$ is 5$\times$3 grid, and $v_{\texttt{i-j}}$ denotes the vertex with row $i$, column $j$.
    Light red indicates $\qfrom$, and blue represents $T$.
    (\emph{Column 1}) first \algname call $\procname(v_{\texttt{4-3}}, \emptyset,\{v_{\texttt{3-2}},v_{\texttt{4-1}},v_{\texttt{4-2}}\})$ ($v_{\texttt{4-3}}$ is marked by yellow). Then $V_S=\{v_{\texttt{1-2}},v_{\texttt{2-1}},v_{\texttt{2-2}},v_{\texttt{2-3}}\}$, and $\textsf{cur}=v_{\texttt{1-2}}$. \procname returns $\reserved = \{0,2,4,6\}$.
    (\emph{Column 2}) \algname call $\procname(v_{\texttt{5-1}}, \reserved,\{v_{\texttt{3-2}},v_{\texttt{4-1}},v_{\texttt{4-2}}\})$. Then $V_S=\emptyset$, and $\procname$ return $R$. The call $\procname$ with $t=v_{\texttt{5-2}}$ also fails. We set $\reserved= \{0,2,4,5,6\}$.
    (\emph{Column 3}) \algname calls $\procname(v_{\texttt{3-1}},\reserved, \emptyset)$. $V_S=\{v_{\texttt{2-1}}\}$, and \procname returns $\reserved=\{0,1,2,4,5,6\}$.
    (\emph{Column 4}) Repeatedly calling \procname, finally \algname returns $\current=[v_{\texttt{2-2}},v_{\texttt{3-1}},v_{\texttt{3-2}},v_{\texttt{3-3}},v_{\texttt{4-2}},v_{\texttt{4-1}},v_{\texttt{4-3}}]$.
    }
    \label{fig:execution}
\end{figure}

\paragraph{Pull process -- Enumerating available vertices.}
To identify a pullable path, \procname firstly identifies the set of vertices of the graph $\indG{\setcurrent\cup \{t\}}$ which can be reached at $t$ without crossing the agent $i\in \reserved$ (line \ref{pull:BFS}). 
BFS conducts this starting from $t$ over $\indG{\setcurrent\cup \{t\} \setminus \setcurrent(\reserved)}$. 
Since the chain of moves of agents forms a single path (similar to the discussion in \Cref{sec:concept}), this BFS operation characterizes the set of vertices that satisfy a necessary condition to be an initial vertex of the path (denoted by $\reachable$). 

Next, \procname computes the set of cut vertices of $\indG{\setcurrent\cup \{t\}}$ (line \ref{pull:DFS}).  
If a vertex $v\in \setcurrent$ is a cut vertex, then the agent on $v$ is forced to stay on $v$ to maintain the connectivity.
Thus, this operation characterizes the set of vertices that must not be chosen as the initial vertex of the path (denoted by $\cutset$).

Consequently, $\reachable\setminus \cutset$ represents the set of vertices that can be chosen as an initial vertex. 
Line \ref{pull:judge} determines whether the procedure \procname succeeds or not.
If the condition holds, \procname\hspace{-3pt}$(t,\reserved,V')$ fails to find the agents to move, since there is no appropriate vertex.
Then procedure \procname returns the same configuration and set of agents.
If not, it is possible to form a new path that does not overlap with the existing path, while preserving connectivity.  

\paragraph{Pull process -- Selection.}
Next, \procname selects the initial vertex of the pulling path.
From the set of vertices $V_S$, \procname selects the one that is currently farthest from any target vertex, that is, \textsf{cur} holds $\min_{y\in T}(\dist(v,y))\le \min_{y\in T}(\dist(\textsf{cur},y))$ for every $v\in V_S$ (line \ref{pull:select}).
Then \procname starts to determine next positions sequentially;
for an agent $i$ such that $\qfrom[i]=\textsf{cur}$, \procname determines $\current[i]$ as $\textsf{next}(\textsf{cur},t)$, which is already computed in line \ref{pull:BFS}.
Then agent $i$ is added to set $\reserved$ of agents already determined next position, and the procedure moves to the next vertex $\textsf{next}(\textsf{cur},t)=\current[i]$.
After performing the decisions iteratively, \procname generates the next configuration and returns a set $R$ of agents already determined their next position.

\paragraph{Top-level procedure (lines \ref{pull:goalornot}--\ref{pull:notg_e}).}
\algname repeatedly applies \procname and updates $\current$.
Doing so, \algname can generate the configuration that approaches $T$, if $\setqfrom\cap T= \emptyset$.

\paragraph{When $Q^{\text{from}}$ overlaps with $T$.}
In this case, we need to care about which vertex to pull first.
\Cref{fig:sortcc} (left) shows such an example.
When \procname first pulls an agent in $P_2$ at step $t$ and $P_0$ at step $t+1$, the iteration causes a livelock.
To resolve this issue, we prioritize ``blocks of agents that have reached $T$,'' i.e., the components of $\indG{\qfrom\cap T}$, and then \algname perform \procname to enlarge the largest one further (\Cref{fig:sortcc} lower right).

\begin{figure}[t]
    \centering
\begin{tikzpicture}[scale=0.5, thick]

\begin{scope}[xshift = -0.7cm, yshift = 0.5cm]

  \draw[black, fill=blue!20]
    (-1,-0.2) -- (6.5,-0.2) -- (6.5,2.5) -- (-1,2.5) -- cycle;

  \draw[black, fill=red!50]
    (-0.5,4) -- (-0.5,1) -- (1,1) -- (1,3) -- (2,3) -- (2,0.75) 
    -- (3.5,0.75) -- (3.5,3) -- (4.5,3) -- (4.5,0.5) -- (6,0.5) -- (6,4) -- cycle;

    \node[ fill=white] at (0.5,0.4) {$T$};
    \node[font = \large] at (-0.6,4.2) {$\qfrom$};
    \node[ fill=white] at (5.25, 1.75) {\scriptsize{$P_0$}};
    \node[ fill=white] at (2.75, 1.75) {\scriptsize{$P_1$}};
    \node[ fill=white] at (0.25, 1.75) {\scriptsize{$P_2$}};

    \draw[black] (-1,2.5)--(6.5,2.5);
\end{scope}
\begin{scope}[xshift = 7.8cm,yshift=3.5cm]

  \draw[black, fill=blue!20]
    (-1,-0.2) -- (6.5,-0.2) -- (6.5,2.5) -- (-1,2.5) -- cycle;

  \draw[black, fill=red!50]
    (-0.5,4) -- (-0.5,0.5) -- (1,0.5) -- (1,3) -- (2,3) -- (2,0.75) 
    -- (3.5,0.75) -- (3.5,3) -- (4.5,3) -- (4.5,1) -- (6,1) -- (6,4) -- cycle;

    \node[ fill=white] at (
    5,0.4) {$T$};
    \node[font = \large] at (-0.7,4.2) {$\current$};
    \node[ fill=white] at (5.25, 1.75) {\scriptsize{$P_0$}};
    \node[ fill=white] at (2.75, 1.75) {\scriptsize{$P_1$}};
    \node[ fill=white] at (0.25, 1.75) {\scriptsize{$P_2$}};

    \draw[black] (-1,2.5)--(6.5,2.5);
    \filldraw[
    draw=black!5,        
    line width=0.6pt,  
    fill=black,      
    fill opacity=0.1   
  ] (-1.5,-0.5) rectangle (7,5);
\end{scope}

\begin{scope}[xshift = 7.8cm, yshift = -2.5cm]

  \draw[black, fill=blue!20]
    (-1,-0.2) -- (6.5,-0.2) -- (6.5,2.5) -- (-1,2.5) -- cycle;

  \draw[black, fill=red!50]
    (-0.2,4) --(-0.2,2.5)--(-0.5,2.5)-- (-0.5,1) -- (1,1)-- (1,2.5) -- (0.7,2.5)--(0.7,3.5) -- (2,3.5) -- (2,2.5)--(1.6,2.5)--(1.6,0.75) 
    -- (3.9,0.75) -- (3.9,2.5)--(3.5,2.5)-- (3.5,3.5) -- (4.5,3.5) --(4.5,2.5)-- (3.9,2.5) -- (3.9,0.2) -- (6.5,0.2) -- (6.5,2.5)-- (6,2.5)-- (6,4) -- cycle;

    \node[ fill=white] at (0.5,0.4) {$T$};
    \node[font = \large] at (-0.8,3.9) {$\current$};
    \node[ fill=white] at (5.25, 1.75) {\scriptsize{$P_0$}};
    \node[ fill=white] at (2.75, 1.75) {\scriptsize{$P_1$}};
    \node[ fill=white] at (0.25, 1.75) {\scriptsize{$P_2$}};

    \draw[black] (-1,2.5)--(6.5,2.5);
\end{scope}
\draw[->] (2,4.8)--(2,6.5)--(6,6.5);
\draw[->] (11,3.5)--(11,2)--(6.5,2);
\draw[->] (2,-0.5)--(2,-1.5)--(6,-1.5);
\node at (9,2.6) {Call in $P_0$};
\node at (4,5.9) {Call in $P_2$};
\node at (4,-2.3) {Call in $P_0$};

\end{tikzpicture}
    \caption{A conceptual illustration of line \ref{pull:sortcc}. (\textbf{Left}) initial configuration. (\textbf{Upper right}) intermediate configuration without prioritization. If there is no proper prioritization, for example, at step $\tau$, $P_2$ pulls the agents of $P_0$, and at the step $\tau+1$, $P_0$ pulls the agents of $P_2$, causing a livelock. (\textbf{Lower right}) itermediate configuration before executing line \ref{pull:notg_s}. In the initial iteration, \algname calls \procname on $\nei{P_0}$, which increases the size of the largest component. Two connected components might merge (e.g., $P_0$ and $P_1$ in the figure), and these are computed as a larger component in the next step.}
    \label{fig:sortcc}
\end{figure}
To be concrete, if $\setqfrom\cap T\ne \emptyset$, we first computes the list of components of $\indG{\setqfrom\cap T}$ (denoted by $\mathcal{P}$). 
Note that even if $\qfrom$ is connected, $\indG{\setqfrom\cap T}$ can consist of multiple components, as illustrated in \Cref{fig:sortcc}.
We next sort the elements of $\mathcal{P}$ by their size of vertex set, and call \procname in decreasing order.
We now observe that, if $\setqfrom\cap T\ne \emptyset$, \algname first calls procedure \procname\hspace{-3pt}$(t,\ast,\ast)$ with $t\in \nei{P_0}$. 
Thus, $\current$ holds that $\indG{\setcurrent\cap T}$ has the larger component than $\indG{\setqfrom\cap T}$ (see \Cref{fig:sortcc}).
Thus, the size of the largest component becomes larger step by step.
Combining this method and process at line \ref{pull:goal_e}, we can prove the convergence.

\subsection{Theoretical Analysis}\label{sec:proof}
We write $\textsf{cur}_{k}$, $\finalk{k}$, $\reachablek{k}$, $\currentk{k}$, $\cutsetk{k}$ and $\reservedk{k}$, where $k$ is the number of calls to the \procname within a single execution of algorithm \ref{alg:PULL} for convinience.
Due to space limitations, we omit the proofs in the paper, and we give some proof sketches for key lemmas. 
See \Cref{sec:omittedproofs} for full proofs.

We first show that our proposed algorithm correctly generates the configuration $\current$ from the input $\qfrom$.
\begin{restatable}{lem}{connectivity}\label{lem: connectivity}
    If $\qfrom$ is connected, then $\current$ is connected, collision-free, and reachable.
\end{restatable}
\begin{proof}
(sketch)
Regarding reachability, it is clear since the decision of $\current[i]$ is made according to $\current[i]\in N(\qfrom[i])\cup \{\qfrom[i]\}$ for every $i\in A$, from line \ref{pull:assign}.
Furthermore, for each agent $i\in A$, once its destination $\current[i]$ is fixed, it is added to the set $R$, so its destination will not be changed.
From here, we show that $k$-th call of \procname generates a connected, collision-free configuration $\currentk{k}$ from $\qfrom_{k}=\currentk{k-1}$.
Upon a successful pull by \procname, each agent on the $(cur,t)$-path moves one step forward. 
Ignoring labels, this is equivalent to moving an agent at $\textsf{cur}$ to $t$, hence $\currentk{k'+1}=(\currentk{k'}\cup \{t\}\}\setminus \{\textsf{cur})$.
Since $\textsf{cur}$ is not a cut vertex of $\indG{\currentk{k'}\cup \{t\}}$, we can conclude $\currentk{k'+1}$ is connected.
Next, there is no conflict within a procedure in line \ref{pull:assign}. Moreover, in another process, we search paths while avoiding agents in $R$ and destinations $\current(R)$ that have already been determined next positions, thus there is no path crossing with an existing one. Therefore, no conflict occurs throughout, and we conclude that $\current$ is conflict-free.
\end{proof}

Next, we show that $\current\ne \qfrom$, that is, at least one \procname call succeeds.
Note that, for every call of \procname, we can assume that the third argument $V'$ is either empty or forms a connected induced subgraph. 

\begin{restatable}{lem}{firstcall}\label{lem: firstcall}
    Let $\finalk{1}$ be the vertex that \procname$(\finalk{1},\emptyset,V')$ is called in step $\tau$. There is an agent $i$ such that $\current[i] = \finalk{1}$ in line \ref{pull:assign}.
\end{restatable}
\begin{proof}(sketch)
To prove the lemma, we must show that the set $V_S$ (line \ref{pull:judge}) is always non-empty for the first call of \procname, i.e., when $\reserved =\emptyset$, since $\reachable = \qfrom$.
The case when $V'=\emptyset$ is clear since every graph has at least two vertices that are not the cut vertex.
Even after removing $t$, at least one vertex remains; hence $V_s$ is nonempty.

Next, we consider when $V'\ne \emptyset$, that is, \procname is called in line \ref{pull: goal_call}.
The proof is conducted by contradiction.
Consider the vertex $w\in \qfrom\setminus V'$ such that $w$ holds that the largest minimum distance $\min_{t\in V'}(\dist(v,t))$ among every vertex $v\in \qfrom\setminus V'$ ($\spadesuit$).
Suppose that $w$ is a cut vertex of $\indG{\qfrom\cup \{t\}}$ for a contradiction.
Then, the graph $\indG{(\qfrom\cup \{t\})\setminus \{w\}}$ consists of more than one component (\Cref{fig:lem2} left).
Let $C_1$ be a component that contains $V'$ (note that, since $\indG{V'}$ is connected, it is contained in a single component), and consider the vertex $u'$ in another component.
Since $w$ is a cut vertex, every path connecting $u'$ and $v'\in V'$ contains $w$ as an internal vertex, and thus the shortest $(u',v')$-path also does.
Therefore, $\dist(u',v')\ge \dist(w',v')+1$ holds (\Cref{fig:lem2} right); a contradiction to the maximality of $w$ ($\spadesuit$), thus $w$ does not cut $\qfrom$.
\end{proof}

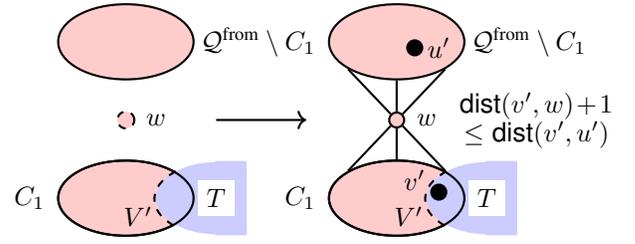
\begin{figure}
    \centering
\begin{tikzpicture}[xscale=0.8, yscale = 0.8, thick]

\node at (2.2,2.6) {$\setqfrom\setminus C_1$};
\node at (-1.6,0) {$C_1$};
\node[fill =red!50,draw, ellipse, minimum width=1.8cm, minimum height=1cm] (C2) at (0,2.6) {};
\node[fill =red!50,draw, ellipse, minimum width=1.8cm, minimum height=1cm] (C1) at (0,0) {};
\begin{scope}
    \clip (-1,-1) rectangle (2,2);
    
        \fill[blue!20] (1.5,0) ++(90:1 and 0.6) arc (90:270:1 and 0.6);
        \fill[draw = blue!20, fill = blue!20] (1.5,-0.6) rectangle (2,0.6);
\end{scope}
\node[draw, ellipse, minimum width=1.8cm, minimum height=1cm] at (0,0) {};
\draw[dashed] ($(C1.center)+(1.5,0)$) ++(132:1 and 0.6) arc (130:230:1 and 0.6);
\node at ($(C1.center)+(0.2,-0.3)$) {$V'$};
\node[fill=white] at ($(C1.center)+(1.5,0)$) {$T$};

\filldraw[fill=red!50, draw=black,dashed] (0,1.3) circle (4pt);
\node at (0.5,1.3) {$w$};

\draw[->, thick] (1.5,1.3) -- (3,1.3);

\begin{scope}[xshift=4.5cm]
  \node at (2.2,2.6) {$\setqfrom\setminus C_1$};
\node at (-1.6,0) {$C_1$};
\node[fill =red!50,draw, ellipse, minimum width=1.8cm, minimum height=1cm] (C2) at (0,2.6) {};
\node[fill =red!50,draw, ellipse, minimum width=1.8cm, minimum height=1cm] (C1) at (0,0) {};
\begin{scope}
    \clip (-1,-1) rectangle (2,2);
    
        \fill[blue!20] (1.5,0) ++(90:1 and 0.6) arc (90:270:1 and 0.6);
        \fill[draw = blue!20, fill = blue!20] (1.5,-0.6) rectangle (2,0.6);
\end{scope}
\node[draw, ellipse, minimum width=1.8cm, minimum height=1cm] at (0,0) {};
\draw[dashed] ($(C1.center)+(1.5,0)$) ++(132:1 and 0.6) arc (130:230:1 and 0.6);
\node at ($(C1.center)+(0.2,-0.3)$) {$V'$};
\node[fill=white] at ($(C1.center)+(1.5,0)$) {$T$};

  \node[circle, draw=black, fill=red!50, inner sep=2pt] (w) at ($ (C1)!0.5!(C2) $) {};
  \node at ($(w)+(0.5,0)$) {$w$};

  \draw (C1.north west) -- (w.south west);
  \draw (C1.north) -- (w.south);
  \draw (C1.north east) -- (w.south east);
  \draw (C2.south west) -- (w.north west);
  \draw (C2.south) -- (w.north);
  \draw (C2.south east) -- (w.north east);

  \node[circle, draw=black, fill=black, inner sep=2pt] (u) at ($(C2.center)+(0.3,-0.1)$) {};
  \node at ($(u)+(0.4,0)$) {$u'$};
  \node[circle, draw=black, fill=black, inner sep=2pt] (v) at ($(C1.center)+(0.7,0.1)$) {};
  \node at ($(v)+(-0.4,0.2)$) {$v'$};

  \node[align=center, text width=2cm] at (2.3,1.3) {$\dist(v',w)+1$\\$\le \dist(v',u') $};

\end{scope}

\end{tikzpicture}
    \caption{A brief explanation of the proof in \Cref{lem: firstcall}.    (\emph{Left}) Graph $\indG{\setqfrom \cup \{\finalk{1}\} \setminus \{w\}}$ is divided into multiple components, shown as two for convenience.
    (\emph{Right}) A graph $\indG{\setqfrom \cup \{\finalk{1}\}}$ has a single component, and every $(v',u')$-paths must pass through $w$ since $w$ is a cut vertex.}
    \label{fig:lem2}
\end{figure}

From \Cref{lem: firstcall}, we can show the following lemmas, which are necessary to guarantee the completeness of our algorithm.
The first case is when $\qfrom$ does not overlap $T$.
As the approach introduced in \Cref{sec:concept}, we show that at least one of the agents that are currently closest to $T$ moves in a direction that reduces its distance to $T$ by one.

\begin{restatable}{lem}{goforward}\label{lem:goforward}
    If $\setqfrom\cap T = \emptyset$, then 
    \begin{align*}
        \min_{(s,t)\in \setcurrent\times T} \dist(s,t)=\min_{(s,t)\in \setqfrom\times T}\dist(s,t)-1
    \end{align*}
\end{restatable}
Then there exists a step $\tau'$ such that $\setconf{\tau'}$ overlaps $T$.
Thus, the condition at line \ref{pull:goal_s} is satisfied at some step.
From here, let $p^{\max}_{\tau}$ be the size of the largest component of $G[\setconf{\tau}\cap T]$.
We finalise the proof by showing that $p^{\max}_{\tau}$ is strictly increasing.
\begin{restatable}{lem}{convergence}\label{lem: convergence}
    If $\setconf{\tau}\cap T \ne \emptyset$ at step $\tau$, $p^{\max}_{\tau+1}\ge p^{\max}_{\tau}+1$.
\end{restatable}

Then there is a step $\tau^\star$, where $p^{max}_{\tau^\star} = n$. At step $\tau^\ast$, it holds that $\qfrom_{\tau^\star} = T$, which implies that a plan $[\setconf{0},\setconf{1},...,\setconf{\tau^{\star}}]$ is valid.
Furthermore, $\tau^\star$ has an upper bound: $\tau'+(n-1)$. This immediately leads to the following. 
\begin{thm}\label{thm: complete}
    \algname is complete for \prbname.
\end{thm}

\begin{prop}\label{prop: upperbound}
The makespan of \prbname is bounded by $\diam(G) + n-1$, where $\diam(G)\coloneqq\max_{(s,t)\in V^2}\dist(s,t)$.
\end{prop}

For the time complexity of \algname, we have the following:
\begin{restatable}{prop}{runtime}\label{prop: runtime}
    \algname runs in $O(\Delta ^2 n^2)$ time, where $\Delta$ is the maximum degree of $G$.
\end{restatable}
\begin{proof}(sketch)
The worst-case running time can be computed from an upper bound on the number of times \procname is invoked and the time required for a single execution of \procname.
    The number of times \procname is invoked is at most $|N[\qfrom]|\le \Delta n$.
The running time of a single execution of \procname is dominated by the BFS, DFS, and is $O(|\qfrom| + |E(\indG{\qfrom})|)= O(\Delta n)$.
Therefore, the running time is upper bounded by $O(\Delta n )\times O(\Delta n)=O(\Delta^2n^2)$.
\end{proof}
Combining \Cref{prop: upperbound,prop: runtime}, we have the following.
\begin{cor}\label{cor:plantime}
    \algname solves \prbname in $O(|V|\cdot\Delta^2 n^2)$ time.
\end{cor}
Note that, the distance $\min_{t\in T}\dist(u,t)$ for every $u\in V(G)$, used in lines \ref{pull:notg_s} and \ref{pull:select}, can be precomputed in $O(|V|+|E|)=O(\Delta |V|)$ time.

\begin{figure}[tbp]
  \begin{minipage}[b]{0.4\columnwidth}
    \centering
    \begin{tikzpicture}[every node/.style={circle, draw, minimum size=4mm, inner sep=0pt}, thick, scale = 0.5]

    \def\k{5}
    \def\l{4}

    \foreach \i in {1,...,\k} {
        \node[fill=white] (s\i) at (\i,2) {};
    }

    \foreach \i in {1,...,\l} {
        \node (d\i) at (3,2-\i) {};
    }

    \foreach \i in {1,...,\k} {
        \node[fill=white] (t\i) at (\i, -3) {};
    }

    \foreach \i in {1,...,\numexpr\k-1} {
        \pgfmathtruncatemacro{\j}{\i+1}
        \draw (s\i) -- (s\j);
    }

    \foreach \i in {1,...,\k} {
        \draw (s\i) -- (d1);
    }

    \foreach \i in {1,...,\numexpr\l-1} {
        \pgfmathtruncatemacro{\j}{\i+1}
        \draw (d\i) -- (d\j);
    }

    \foreach \i in {1,...,\k} {
        \draw (d\l) -- (t\i);
    }

    \foreach \i in {1,...,\numexpr\k-1} {
        \pgfmathtruncatemacro{\j}{\i+1}
        \draw (t\i) -- (t\j);
        }

    \foreach \i in {1,...,\k} {
        \node[draw =blue, very thick, fill=blue!5, minimum size=2.5mm] (t\i) at (\i, -3) {};
    }
     \foreach \i in {1,...,\k} {
        \node[draw =red, very thick, fill=red, minimum size=2.5mm] (t\i) at (\i, 2) {};
    }

    \draw[decorate, decoration={brace, amplitude=6pt, raise = 3pt}, thick]
  (s1.north west) -- (s5.north east)
  node[midway, yshift=15pt, rectangle, draw=none, fill=white, inner sep=2pt]
  {$k$ vertices};

  \draw[decorate, decoration={brace, amplitude=6pt, raise = 5pt}, thick]
  (d\l.south west) -- (d1.north west)
  node[midway, xshift=-13pt, rectangle, draw=none, fill=white, inner sep=2pt, anchor=east]
  {$\ell$ vertices};
  \node[draw=none] at (-0.5,3) {(a)};

\end{tikzpicture}
  \end{minipage}
  \hspace{0.05\columnwidth}
  \begin{minipage}[b]{0.42\columnwidth}
    \centering
    \begin{tikzpicture}[scale = 0.6, thick]
\node[circle, draw, minimum size=4mm] (v11) at (1, -1) {};
\node[circle, draw, minimum size=4mm] (v12) at (2, -1) {};
\node[circle, draw, minimum size=4mm] (v13) at (3, -1) {};

\node[circle, draw, minimum size=4mm] (v21) at (1, -2) {};
\node[circle, draw, minimum size=4mm] (v22) at (2, -2) {};
\node[circle, draw, minimum size=4mm] (v23) at (3, -2) {};

\node[circle, draw, minimum size=4mm] (v31) at (1, -3) {};
\node[circle, draw, minimum size=4mm] (v32) at (2, -3) {};
\node[circle, draw, minimum size=4mm] (v33) at (3, -3) {};

\node[circle, draw, minimum size=4mm] (v41) at (1, -4) {};
\node[circle, draw, minimum size=4mm] (v42) at (2, -4) {};
\node[circle, draw, minimum size=4mm] (v43) at (3, -4) {};

\node[circle, draw, minimum size=4mm] (v51) at (1, -5) {};
\node[circle, draw, minimum size=4mm] (v52) at (2, -5) {};
\node[circle, draw, minimum size=4mm] (v53) at (3, -5) {};

\draw (v11) -- (v12) -- (v13);
\draw (v21) -- (v22) -- (v23);
\draw (v31) -- (v32) -- (v33);
\draw (v41) -- (v42) -- (v43);
\draw (v51) -- (v52) -- (v53);

\draw (v11) -- (v21) -- (v31) -- (v41) -- (v51);
\draw (v12) -- (v22) -- (v32) -- (v42) -- (v52);
\draw (v13) -- (v23) -- (v33) -- (v43) -- (v53);

\node[draw=red, very thick, fill=red, minimum size=1.5mm, circle] at (v11) {};
\node[draw=red, very thick, fill=red, minimum size=2.5mm, circle] at (v21) {};
\node[draw=red, very thick, fill=red, minimum size=2.5mm, circle] at (v31) {};
\node[draw=red, very thick, fill=red, minimum size=2.5mm, circle] at (v41) {};
\node[draw=red, very thick, fill=red, minimum size=2.5mm, circle] at (v51) {};

\node[draw=blue, very thick, fill=blue!5, minimum size=2.5mm, circle] at (v13) {};
\node[draw=blue, very thick, fill=blue!5, minimum size=2.5mm, circle] at (v23) {};
\node[draw=blue, very thick, fill=blue!5, minimum size=2.5mm, circle] at (v33) {};
\node[draw=blue, very thick, fill=blue!5, minimum size=2.5mm, circle] at (v43) {};
\node[draw=blue, very thick, fill=blue!5, minimum size=2.5mm, circle] at (v53) {};

\draw[decorate, decoration={brace, amplitude=6pt, mirror, raise = 5pt}, thick]
  (v11.north west) -- (v51.south west)
  node[midway, xshift=-13pt, rectangle, draw=none, fill=white, inner sep=2pt, anchor=east]
  {$k$ vertices};
\node[draw=none] at (-0.5,-0.2) {(b)};
\end{tikzpicture}
  \end{minipage}
  \caption{(a) An example of a tight instance: the optimal makespan is $\diam(G)+n-1$. (b) An example of a adversarial instance: \algname outputs plan with $O(n)$ makespan, while the optimal makespan is two.}
  \label{fig: instance}
\end{figure}
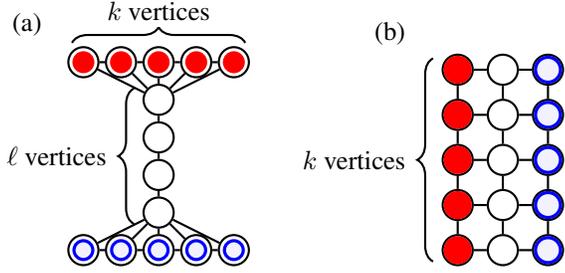

\paragraph{Tightness, Adversarial instances.}

The makespan upper bound established is tight; an example is shown in \Cref{fig: instance}(a).
Furthermore, we show that for some instances, \algname outputs a plan with a makespan of $ O(n)$, even when the optimal makespan is two (see \Cref{fig: instance} (b)). 

\section{Empirical Analysis}
\begin{figure*}[t]
    \centering
    
    \begin{tikzpicture}
    \node[anchor=north west, inner sep=0] (image) at (0,5.8) {
        \includegraphics[width=\linewidth]{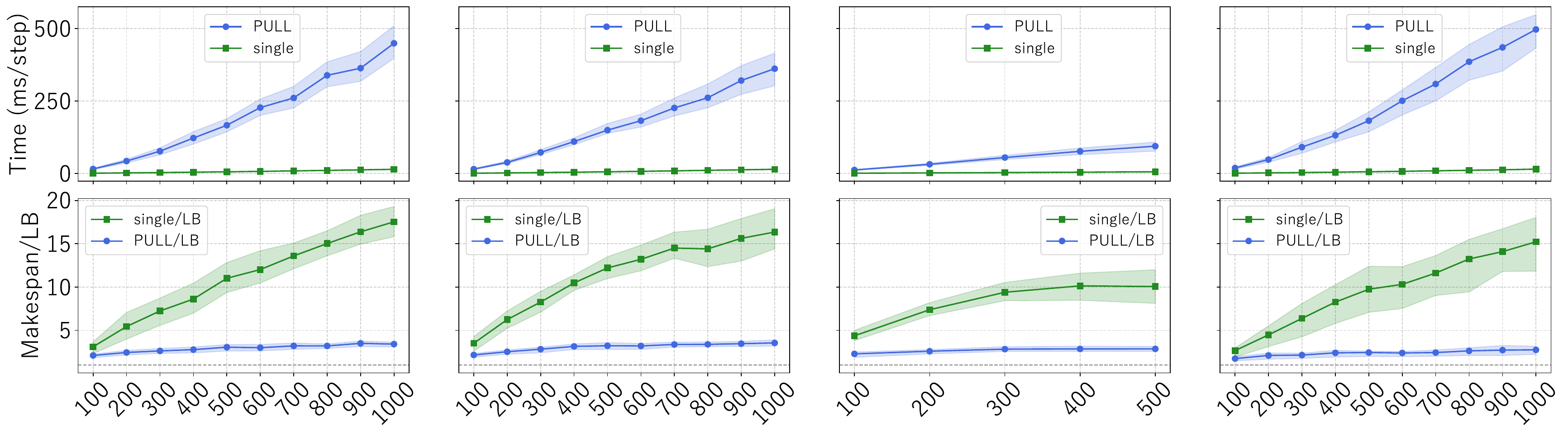}
    };
    \node[font =\small] at (2.8,5.9) {\emph{random-64-64-20}};
    \node[anchor=south west, inner sep=0](image) at (1.1,4.6) {\includegraphics[width = 0.055\linewidth]{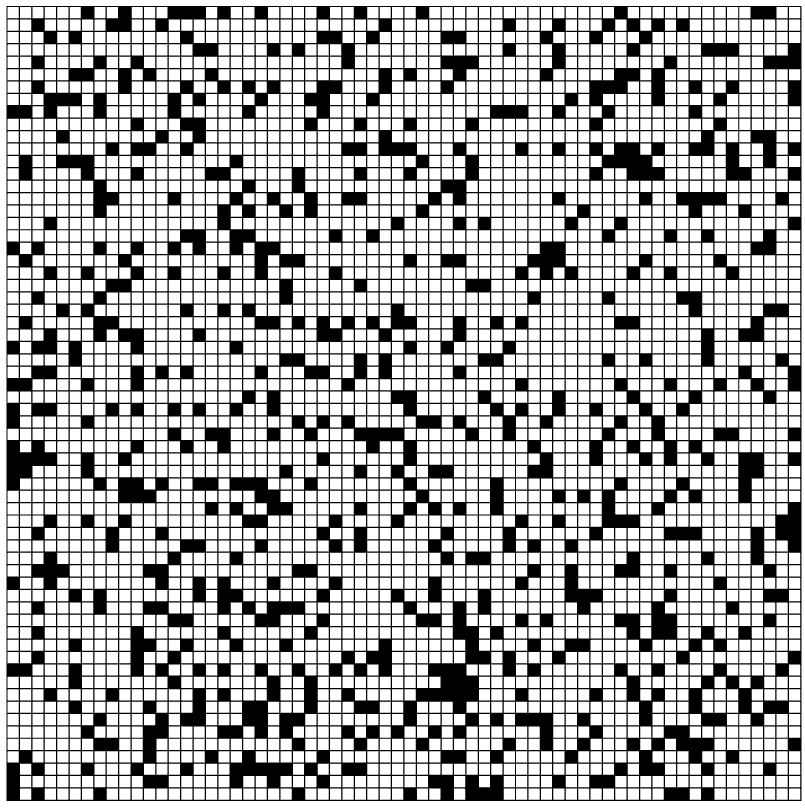}};
    \node[font =\small] at (7.1,5.9) {\emph{random-48-48-20}};
    \node[anchor=south west, inner sep=0](image) at (5.4,4.6) {\includegraphics[width = 0.055\linewidth]{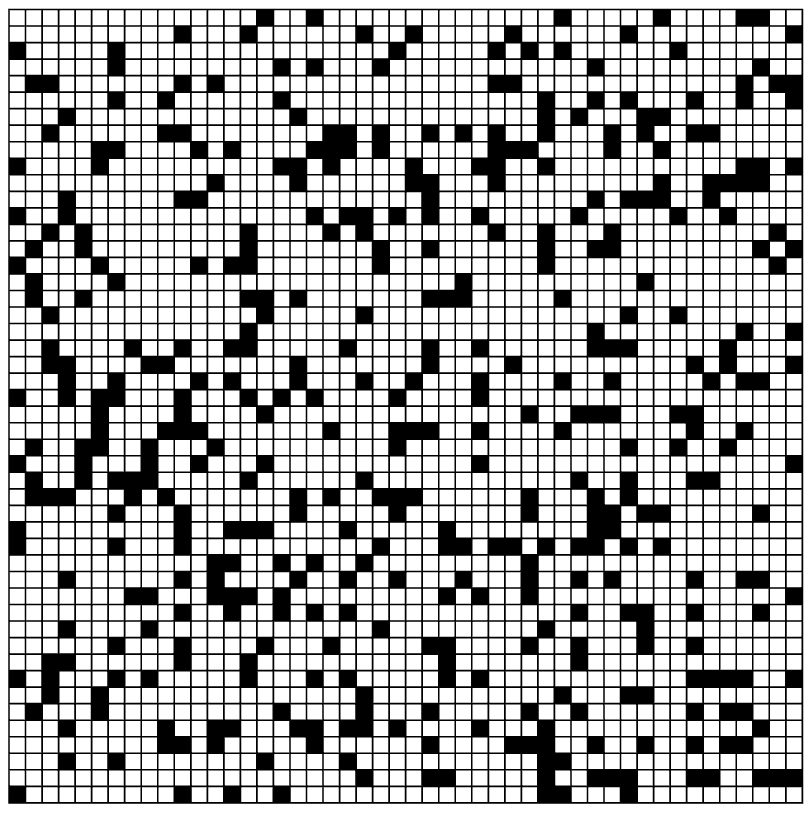}};
    \node[font =\small] at (11.5,5.9) {\emph{random-32-32-20}};
    \node[anchor=south west, inner sep=0](image) at (9.8,4.6) {\includegraphics[width = 0.055\linewidth]{figs/random-32-32-20.png}};
    \node[font =\small] at (15.7,5.9) {\emph{warehouse-10-20-10-2-2}};
    \node[anchor=south west, inner sep=0](image) at (14,5) {\includegraphics[width = 0.07\linewidth]{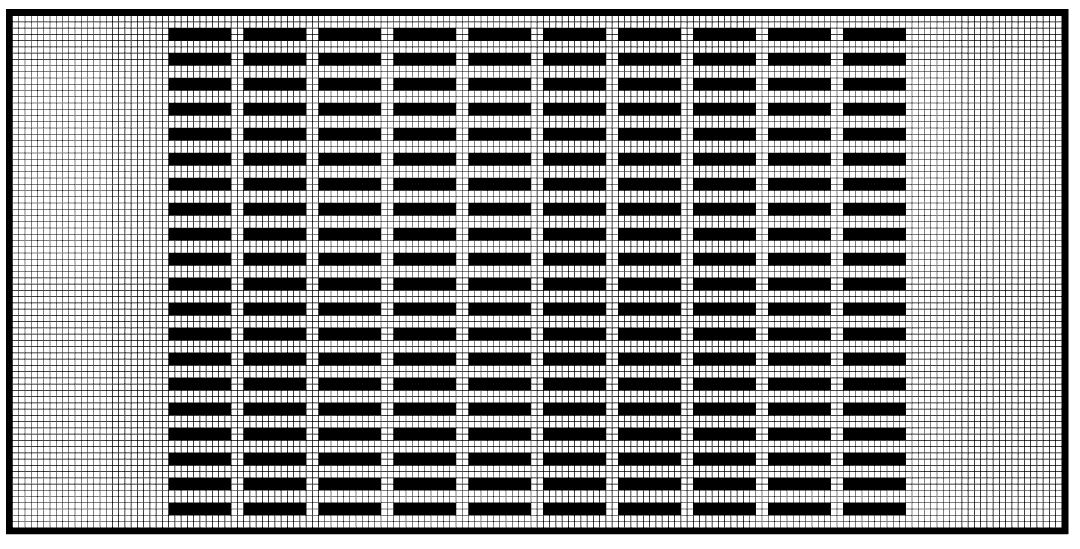}};
    \node[font =\small] at (9.25,0.75) {The number of agents};
    \end{tikzpicture}
    \caption{Empirical results for time and makespan/LB across four maps. Blue represents \Cref{alg:PULL}, and green represents the algorithm in \Cref{sec:concept}. Lines plot the average values, and the filled areas indicate the interquartile range.}
    \label{fig:exam_bas}
\end{figure*}
This section presents an empirical evaluation of \algname. 
Specifically, we conduct the following two: a runtime evaluation against the optimal \ilp algorithm, and an evaluation of the running time and solution quality (i.e., makespan) of \algname against a simple solution on a variety of benchmark maps.
The algorithm is coded in Python, and the experiments were run on a Mini PC with Intel Core i9-13900H 2.6GHz CPU and 32GB RAM.

\subsection{Comparison between \ilp and \algname}
As shown in \Cref{tab:ILP}, the ILP-based algorithm is effective for small instances, but its runtime becomes prohibitively large for larger ones. 
In contrast, \algname finds a solution within the time limit for all instances in ``empty-8-8,'' and on the instances solvable by \ilp, its average running time was almost 10 times faster in the worst. 
In \ilp, we add $O(|V|)$ constraints when a step is incremented, thus the runtime increases by a factor of $O(2^{|V|})$.
In contrast, \algname computes a succeeding configuration with an additional time overhead of $O(\Delta^2 n^2)$, resulting in rapid planning.
Moreover, the makespan \emph{suboptimality}, defined as $\text{makespan}/\text{optimum}$, is at worst on average 1.7, indicating that it outputs practically useful solutions.
For the detailed runtime analysis, see \Cref{sec:ILPapp}.

\subsection{\algname on Large-scale Problem}
We now demonstrate the effectiveness of \algname on large-scale instances.
For our evaluation, we use the \emph{random-32-32-20}, \emph{random-64-64-20}, and \emph{warehouse-10-20-10-2-2} in \mapf benchmarks~\cite{Stern19}, and randomly generated \emph{random-48-48-20}.
We prepare instances with a varying number of agents from 100 to 1,000 in increments of 100, by generating two random connected subgraphs ($S$, $T$).
The evaluation of \ilp is omitted here due to the lack of scalability.

As finding the optimum is hard for these instances, we benchmark our solution quality against a trivial lower bound from the value of bottleneck matching (see~\cite{Matching:GabowT88} for details), which is denoted by LB. 
This represents the plan length required to transform $S$ into $T$, disregarding any conflicts and connectivity.
Note that \algname is the first suboptimal algorithm that addresses \prbname -- no prior methods available. Thus, we used a simple solution in the \Cref{sec:concept} (denoted as \textsf{single}),
the one that calls \procname only once at each step,
as a baseline for solution quality.

\Cref{fig:exam_bas} displays the empirical results for the time and makespan/LB for the two algorithms \algname and \textsf{single}, based on 100 instances per setting.
\Cref{app:resulttable} shows the table of detailed results.
The main observation is as follows:
\begin{itemize}
    \item For each map, the average computational time of \algname increases with the number of agents; however, its growth is more gradual than $O(\Delta^2n^2)=O(n^2)$. It implies that the actual number of calls of \procname is $o(\Delta n)$ in some practical situations. This suggests the algorithm is scalable to even larger instances, enabling it to find solutions within a practical time.
    \item The average suboptimality increases for both algorithms as the density (i.e., $n/|V|$) increases; however, \algname exhibits a considerably more suppressed increase. Specifically, across all maps, when $n=500$, the makespan of \algname is 0.3 times or less than that of the simple solution. Moreover, in sparse conditions, \algname outputs solutions closer to the lower bound than \textsf{single}. We can conclude that \algname produces solutions that are comparatively close to the LB when compared to the simple solution, leading to more efficient operations.
\end{itemize}

\paragraph{Sensitivity to the map size.}
We conducted an additional experiment to demonstrate that the computational time is largely independent of the number of vertices $|V|$. 
Theoretically, the execution time depends solely on $n$ in the \mapf benchmark, while the map size is only involved in distance calculations during preprocessing. 
Therefore, the average time is expected to be mostly the same when $n$ is fixed.
\Cref{fig:exam_mapsize} presents the running time and makespan/LB for a fixed $n=100$,
on a set of random maps of increasing size (e.g., from 16x16 to 64x64, incrementing width by 8).
Note that all of the maps (except \emph{random-32-32-20} and \emph{random-64-64-20}) are randomly generated with 20\% obstacles.
We observe that $|V|$ has a small impact on \algname's running time. 
We can also see that makespan/LB decreases as density gets smaller, consistent with the previous observation.

\begin{figure}[ht]
    \centering
    \includegraphics[width=\linewidth]{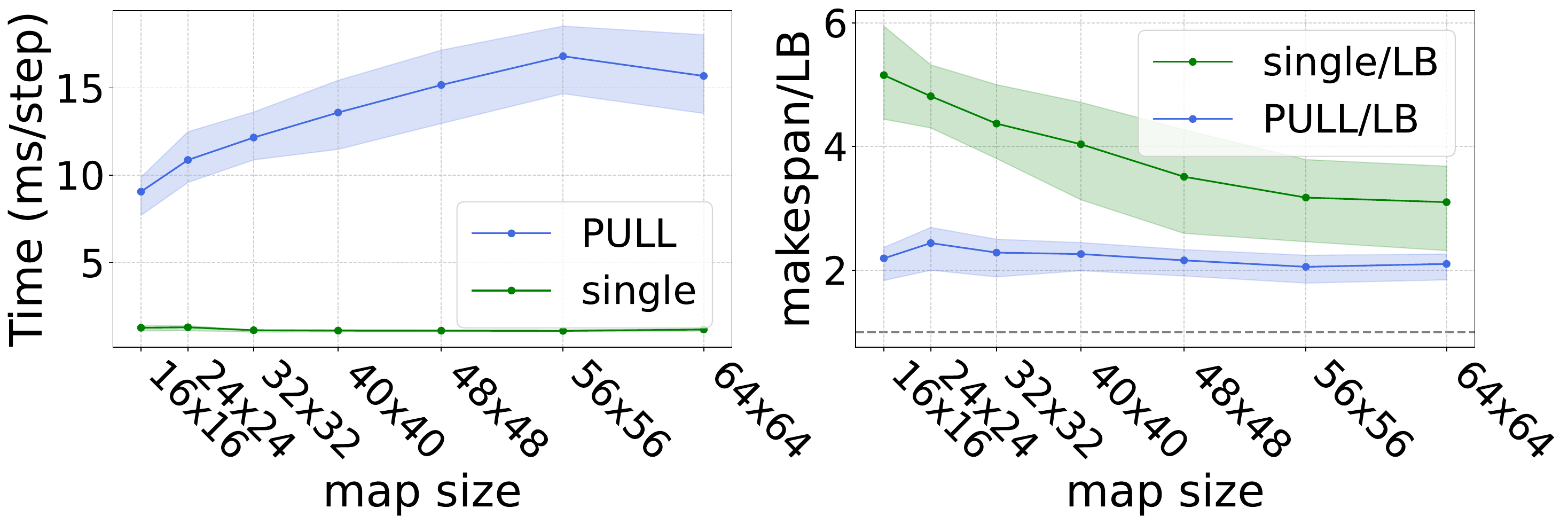}
    \caption{\Cref{alg:PULL} in several maps with $n=100$. }
    \label{fig:exam_mapsize}
\end{figure}

\section{Conclusion and discussion}
This paper tackled \prbname, a variant of \mapf that introduces a connectivity constraint on the entire agent configuration. This makes conventional \mapf algorithms inapplicable. Moreover, there are no algorithms readily applicable to \prbname in practical settings.
We first formulate \prbname (also Unlabeled \mapf) for \ilp, and obtain an optimal algorithm. Next, we introduced a rule-based algorithm \algname, that empirically finds solutions superior to vanilla approaches.
This work contributes to the literature on \mapf with additional constraints, a promising area for future research. 
The discussion is as follows.
\paragraph{Algorithm improvement?}
Several directions exist to improve \algname. 
For example, one might consider implementations using recursive functions, similar to PIBT~\cite{PIBT22}, or reduction to flow-like problems.
Our early attempts observed that neither of these two methods can efficiently handle connectivity, making a significant algorithmic speedup challenging. 
Nevertheless, the latter approach---specifically the escape problem~\cite{Cormen22}---might be applicable to refinement algorithms.

\paragraph{Anytime algorithm.}
As noted in \Cref{tab:MAPFwork}, we also developed an eventually optimal algorithm for \prbname that integrates \algname with a search-based \mapf algorithm \lacam~\cite{LaCAM23}.
We use \algname as a configuration generator and perform an anytime search over configurations, which is guaranteed to converge to an optimal solution eventually.
While this approach finds an optimal solution for small instances, the refinement effect for larger instances is almost nothing in practice.
Developing a sophisticated anytime scheme for \prbname remains an open question.

\bibliography{ref-macro, aaai2026}

@inproceedings{Li20,
  author    = "Jiaoyang Li and Kexuan Sun and Hang Ma and Ariel Felner and T. K. Satish Kumar and Sven Koenig",
  title     = "Moving Agents in Formation in Congested Environments",
  booktitle = aamas,
  pages     = "726-734",
  year      = "2020",
  doi       = "",
}

@article{Fekete21,
author = {Fekete, S\'{a}ndor P. and Keldenich, Phillip and Kosfeld, Ramin and Rieck, Christian and Scheffer, Christian},
title = {Connected coordinated motion planning with bounded stretch},
year = {2023},
issue_date = {Dec 2023},
publisher = {Kluwer Academic Publishers},
address = {USA},
volume = {37},
number = {2},
issn = {1387-2532},
url = {https://doi.org/10.1007/s10458-023-09626-5},
doi = {10.1007/s10458-023-09626-5},
journal = {Autonomous Agents and Multi-Agent Systems},
month = oct,
numpages = {29}
}

@InProceedings{Stern19,
  title={Multi-Agent Pathfinding: Definitions, Variants, and Benchmarks},
  author={Roni Stern and Nathan R. Sturtevant and Ariel Felner and Sven Koenig and Hang Ma and Thayne T. Walker and Jiaoyang Li and Dor Atzmon and Liron Cohen and T. K. Satish Kumar and Eli Boyarski and Roman Bartak},
  booktitle=socs,
  year={2019},
  pages={151--158}
}

@InProceedings{Yu13unlabeled,
author="Yu, Jingjin
and LaValle, Steven M.",
title="Multi-agent Path Planning and Network Flow",
booktitle="Algorithmic Foundations of Robotics X",
year="2013",
pages="157--173",
}

@InProceedings{Hinnenthal24,
  author =	{Hinnenthal, Kristian and Liedtke, David and Scheideler, Christian},
  year = {2024},
  title =	{{Efficient Shape Formation by 3D Hybrid Programmable Matter: An Algorithm for Low Diameter Intermediate Structures}},
  booktitle =	sand,
  pages =	{15:1--15:20}
}

@InProceedings{Kostitsyna23,
  author =	{Kostitsyna, Irina and Peters, Tom and Speckmann, Bettina},
  title =	{{Fast Reconfiguration for Programmable Matter}},
  booktitle =	disc,
  pages =	{27:1--27:21},
  year =	{2023},
  volume =	{281},
  URL =		{https://drops.dagstuhl.de/entities/document/10.4230/LIPIcs.DISC.2023.27},
  URN =		{urn:nbn:de:0030-drops-191533},
  doi =		{10.4230/LIPIcs.DISC.2023.27}
}

@article{TSWAP23,
title = {Solving simultaneous target assignment and path planning efficiently with time-independent execution},
journal = {Artificial Intelligence},
volume = {321},
pages = {103946},
year = {2023},
issn = {0004-3702},
doi = {https://doi.org/10.1016/j.artint.2023.103946},
url = {https://www.sciencedirect.com/science/article/pii/S0004370223000929},
author = {Keisuke Okumura and Xavier Défago}
}

@InProceedings{Kristan25,
author="K{\v{r}}i{\v{s}}{\v{t}}an, Jan Maty{\'a}{\v{s}}
and Svoboda, Jakub",
title="Reconfiguration Using Generalized Token Jumping",
booktitle= walcom,
year="2025",
pages="244--265",
}

@book{Cormen22,
  title={Introduction to Algorithms, fourth edition},
  author={Cormen, T.H. and Leiserson, C.E. and Rivest, R.L. and Stein, C.},
  isbn={9780262046305},
  lccn={2021037260},
  url={https://books.google.co.jp/books?id=HOJyzgEACAAJ},
  year={2022},
  publisher={MIT Press}
}

@article{PIBT22,
title = {Priority inheritance with backtracking for iterative multi-agent path finding},
journal = {Artificial Intelligence},
volume = {310},
pages = {103752},
year = {2022},
issn = {0004-3702},
doi = {https://doi.org/10.1016/j.artint.2022.103752},
url = {https://www.sciencedirect.com/science/article/pii/S0004370222000923},
author = {Keisuke Okumura and Manao Machida and Xavier Défago and Yasumasa Tamura}
}

@inproceedings{LaCAM23,
  title     = {Improving LaCAM for Scalable Eventually Optimal Multi-Agent Pathfinding},
  author    = {Okumura, Keisuke},
  booktitle = ijcai,
  pages     = {243--251},
  year      = {2023},
  month     = {8},
  doi       = {10.24963/ijcai.2023/28},
  url       = {https://doi.org/10.24963/ijcai.2023/28},
}

@article{Sharon15,
title = {Conflict-based search for optimal multi-agent pathfinding},
journal = {Artificial Intelligence},
volume = {219},
pages = {40-66},
year = {2015},
issn = {0004-3702},
doi = {https://doi.org/10.1016/j.artint.2014.11.006},
url = {https://www.sciencedirect.com/science/article/pii/S0004370214001386},
author = {Guni Sharon and Roni Stern and Ariel Felner and Nathan R. Sturtevant},
}

@inproceedings{Roger12,
  author       = {Gabriele R{\"{o}}ger and
                  Malte Helmert},
  title        = {Non-Optimal Multi-Agent Pathfinding Is Solved (Since 1984)},
  booktitle    = {Multiagent Pathfinding, Papers from the 2012 {AAAI} Workshop},
  series       = {{AAAI} Technical Report},
  volume       = {{WS-12-10}},
  year         = {2012},
  url          = {http://www.aaai.org/ocs/index.php/WS/AAAIW12/paper/view/5206},
  bibsource    = {dblp computer science bibliography, https://dblp.org}
}

@inproceedings{Yu13hard,
author = {Yu, Jingjin and LaValle, Steven M.},
title = {Structure and intractability of optimal multi-robot path planning on graphs},
year = {2013},
booktitle = aaai,
pages = {1443–1449},
numpages = {7},
location = {Bellevue, Washington}
}

@article{Li22, title={MAPF-LNS2: Fast Repairing for Multi-Agent Path Finding via Large Neighborhood Search}, volume={36}, number={9}, journal=aaai, author={Li, Jiaoyang and Chen, Zhe and Harabor, Daniel and Stuckey, Peter J. and Koenig, Sven}, year={2022}, month={Jun.}, pages={10256-10265} }

@article{Diaz21,
title = {Platooning of connected automated vehicles on freeways: a bird’s eye view},
journal = {Transportation Research Procedia},
volume = {58},
pages = {479-486},
year = {2021},
issn = {2352-1465},
doi = {https://doi.org/10.1016/j.trpro.2021.11.064},
url = {https://www.sciencedirect.com/science/article/pii/S2352146521008231},
author = {Margarita Martínez-Díaz and Christelle Al-Haddad and Francesc Soriguera and Constantinos Antoniou}
}

@article{Wagner15,
title = {Subdimensional expansion for multirobot path planning},
journal = {Artificial Intelligence},
volume = {219},
pages = {1-24},
year = {2015},
issn = {0004-3702},
doi = {https://doi.org/10.1016/j.artint.2014.11.001},
url = {https://www.sciencedirect.com/science/article/pii/S0004370214001271},
author = {Glenn Wagner and Howie Choset}
}

@misc{LF25,
      title={LF: Online Multi-Robot Path Planning Meets Optimal Trajectory Control}, 
      author={Ajay Shankar and Keisuke Okumura and Amanda Prorok},
      year={2025},
      eprint={2507.11464},
      archivePrefix={arXiv},
      primaryClass={cs.RO},
      url={https://arxiv.org/abs/2507.11464}, 
}

@InProceedings{Tateo2018, title={Multiagent Connected Path Planning: PSPACE-Completeness and How to Deal With It}, volume={32}, url={https://ojs.aaai.org/index.php/AAAI/article/view/11587}, DOI={10.1609/aaai.v32i1.11587}, booktitle=aaai, author={Tateo, Davide and Banfi, Jacopo and Riva, Alessandro and Amigoni, Francesco and Bonarini, Andrea}, year={2018}, month={Apr.} }

@InProceedings{Conrad07,
author="Conrad, Jon
and Gomes, Carla P.
and van Hoeve, Willem-Jan
and Sabharwal, Ashish
and Suter, Jordan",
title="Connections in Networks: Hardness of Feasibility Versus Optimality",
booktitle="Proceedings of Integration of AI and OR Techniques in Constraint Programming for Combinatorial Optimization Problems ({CPAIOR})",
year="2007",
publisher="Springer Berlin Heidelberg",
address="Berlin, Heidelberg",
pages="16--28",
}

@article{Matching:GabowT88,
  author       = {Harold N. Gabow and
                  Robert Endre Tarjan},
  title        = {Algorithms for Two Bottleneck Optimization Problems},
  journal      = {J. Algorithms},
  volume       = {9},
  number       = {3},
  pages        = {411--417},
  year         = {1988},
  url          = {https://doi.org/10.1016/0196-6774(88)90031-4},
  doi          = {10.1016/0196-6774(88)90031-4},
  bibsource    = {dblp computer science bibliography, https://dblp.org}
}

@string{science={Science}}

@string{aaai={Proceedings of AAAI Conference on Artificial Intelligence (AAAI)}}

@string{ijcai={Proceedings of International Joint Conference on Artificial Intelligence (IJCAI)}}

@string{aamas={Proceedings of International Joint Conference on Autonomous Agents \& Multiagent Systems (AAMAS)}}

@string{socs={Proceedings of Annual Symposium on Combinatorial Search (SoCS)}}

@string{sand={Proceedings of Symposium on Algorithmic Foundations of Dynamic Networks (SAND)}}

@string{disc={Proceedings of the International Symposium on Distributed Computing (DISC)}}

@string{walcom={Proceedings of International Conference and Workshops on Algorithms and Computation (WALCOM)}}
\newpage
\appendix
\section{Omitted proofs of \Cref{sec:proof}}\label{sec:omittedproofs}

\subsection{Proof of \Cref{lem: connectivity}}
\connectivity*
\begin{proof}
We prove the lemma by induction on $k$, the number of calls to the \procname. 
For the base case, we have $\currentk{0} = \qfrom$, which is connected and collision-free by the assumption.

Assume that $\currentk{k}$ is connected and collision-free. When $V_S = \emptyset$ at line \ref{pull:judge}, then $\currentk{k}=\currentk{k+1}$ since $\reservedk{k} = \reservedk{k+1}$. Then $\currentk{k+1}$ is also connected.
Consider the case when $(\reachablek{k}\setminus \cutsetk{k})\setminus (V'\cup \{t\}) \ne \emptyset$.
We can see that the set of occupied vertices after the move is $\currentk{k+1} = (\currentk{k} \setminus \{\textsf{cur}\}) \cup \{t\}$. 
The procedure \procname selects $\textsf{cur}$ from a set of vertices $V_S$ that excludes any cut vertices of the graph $G[\currentk{k} \cup \{t\}]$. Since removing a non-cut vertex from a connected graph preserves its connectivity, the graph $\indG{\currentk{k+1}}$ is also connected.
Thus, we can prove that $\indG{\current}$ is connected by induction.

Next, we show that there is no vertex conflict and swap conflict in the $(k+1)$-th call.
First, assume that there is a vertex conflict, i.e., there is a vertex $v$ in $(\textsf{cur},t)$-path such that $v\in\current(\reservedk{k})$.
If $v=t$, $t$ was already occupied by another agent, since $t\notin \setqfrom$. However, this contradicts the statement in line \ref{pull: goal_call}, which excludes the occupied vertices.
If $v\ne t$, it holds that $v\in \setcurrent(\reservedk{k})$ and $v\in \reachablek{k+1}\subseteq \setcurrent\setminus\setcurrent(\reserved)$, by the construction of $\reachablek{k}$.
This fact immidiately contradicts that $\setcurrent(\reservedk{k})\cap (\setcurrent\setminus\setcurrent(\reserved))=\emptyset$. Therefore, no two agents can occupy the same vertex in $\currentk{k+1}$.
Second, assume that there is a swap conflict, that is, there is a pair of agents $i,j\in R$ such that $\currentk{k+1}[i]=\qfrom[j] \wedge \currentk{k+1}[j]=\qfrom[i]$. 
Note that there is no collision occurs in $\currentk{k}$.
Since both agents are moved by separate \procname calls within the same one-step function, one must be moved before the other. Assume without loss of generality that agent $i$ is moved before agent $j$.
Note that there is no collision in $\currentk{k}$, we can assume that the move of agent $j$ occurs in $(k+1)$-th call, by induction hypothesis.
Now consider the configuration $\currentk{k}$, which is the configuration that right before agent $j$ is moved. In $\currentk{k}$ agent $j$ has not yet moved, thus $\currentk{k}[j]=\qfrom_k[j]$ holds. Furthermore, agent $i$ has already determined $\current[i]$. By our initial assumption of a swap conflict, we know that $\currentk{k+1}[i]=\qfrom[j]$ and $\currentk{k+1}[i]=\currentk{k}[i]$, we have $\currentk{k}[i]=\qfrom[j]$.
Then we have $\currentk{k}[j]=\currentk{k}[i]$ at the $k$-th call, which immediately contradicts to the fact that there is no vertex conflict in $\currentk{k}$.

This claim holds for all $k\in \mathbb{N}$, then $\currentk{\tau}$ is connected and collision-free.

To the last, it is clear that $\current[i]\in \nei{\setqfrom}\cup \setqfrom$ for all the agents, thus $\current_{\tau}$ is a reachable configuration from $\qfrom$. This completes the proof.
\end{proof}
\subsection{Proof of \Cref{lem: firstcall}}
\firstcall*
\begin{proof}
    To prove the lemma, we must show that this set $V_S$ is always non-empty when $\reserved =\emptyset$. This is equivalent to showing that there is at least one vertex $v\notin \cutsetk{1}$ in the set of $(\setqfrom \cup \{\finalk{1}\}) \setminus (V'\cup \{\finalk{1}\}) = \setqfrom \setminus V'$.

First, we show the case when $V'=\emptyset$. Any connected graph with two or more vertices has at least two vertices that are not cut vertices. Thus $|\reachable\setminus \cutset|\ge 2$ holds for the graph $G[\setqfrom \cup \{\finalk{1}\}]$. Since $\indG{\setqfrom\cup \{\finalk{1}\}}$ is connected, $\finalk{1}$ cannot be the cut vertex of $G[\setqfrom \cup \{\finalk{1}\}]$. Thus $|V_S|\ge 2-1= 1$ holds when $V'=\emptyset$.

We move to the case where $V'\ne \emptyset$ and $\indG{V'}$ is connected.
Note that $\finalk{1}\in \nei{V'}$ when \algname calls the procedure, thus $\indG{V'\cup \{\finalk{1}\}}$ is connected.
We assume that every vertex in $\setqfrom\setminus V'$ is a cut vertex of $\indG{\setqfrom\cup\{\finalk{1}\}}$ for a contradiction. 
Let $w$ be a vertex that holds $w =\argmax (\min_{y\in T\cup \{\finalk{1}\}}\dist(u,y))$.
From the assumption, $w$ is a cut vertex of $\indG{\setqfrom\cup \{\finalk{1}\}}$.
Thus, $G' = \indG{\setqfrom\cup\{\finalk{1}\}\setminus\{w\}}$ consists of at least two components, and all vertices in $V'\cup \{\finalk{1}\}$ are contained in the same component.

Let $C_1$ be a vertex set that $V'\cup \{\finalk{1}\}\subseteq C_1$ and $\indG{C_1}$ is a component of $G'$, and $u$ be a vertex $u'\in V\setminus C_1$ (\Cref{fig:lem2} left).
Since $w$ is a cut vertex and $v'\in V'$ and $u'$ is in different component in $G'$ for every vertex $v'\in V'$, every $(v',u')$-path must pass through $w$. The same applies to shortest paths.
Thus, $u'$ holds that $ \min_{y\in T\cup \{\finalk{1}\}}(\dist(u',y))\ge \min_{y\in T\cup \{\finalk{1}\}}(\dist(w,y))+1$ (\Cref{fig:lem2} right).
This leads a contradiction to the assumption that $w =\argmax (\min_{y\in T\cup \{\finalk{1}\}}\dist(u,y))$.
Therefore there is at least one vertex in $(\reachablek{1}\setminus \cutsetk{1})\setminus (V'\cup \{\finalk{1}\})$.

Then the call \procname$(\finalk{1},\emptyset,V')$ invokes line \ref{pull:assign}. Thus, there is an agent $i$ such that $ \current[i]=\finalk{1}$, completing the proof.
\end{proof}
\subsection{Proof of \Cref{lem:goforward}}
\goforward*
\begin{proof}
    From the assumption, note that lines \ref{pull:goal_s}--\ref{pull:goal_e} are not executed. 
    Consider the vertex $u_0$, the first vertex in the vertex set sorted on line \ref{pull:notg_s}.
    We can see that $\min_{t\in T}\dist(u_0,t)=\min_{(s,t)\in \setqfrom\times T}\dist(s,t)-1$. Note that $\reserved$ is initialized with $\emptyset$ in line \ref{pull:init}. We can apply \Cref{lem: firstcall}, and we can see that there is an agent $a_i$ such that $\current[i]=u_0$.
\end{proof}
\subsection{Proof of \Cref{lem: convergence}}
\convergence*
\begin{proof}
    Let $\mathcal{P}$ be a list of components in $G[\setconf{\tau}\cap T]$, sorted according to their size. Note that $\qfrom = \conf{\tau}$ and $\current=\conf{\tau+1}$.
    First, we show $\mathcal{P}[0]\subseteq \currentk{k}$ after the $k$-th call of procedure \procname for every $k$ call by induction. Note that $\mathcal{P}[0]\subseteq \currentk{0}$ since $\currentk{0}=\qfrom$.
    Assume that $\mathcal{P}[0]\subseteq \currentk{k'}$ holds for some $k'$. First we consider the case where $\finalk{k'+1}\in N(\mathcal{P}[0])$. Then it clearly holds that $\currentk{k+1} = \currentk{k}\cup \{\finalk{k'+1}\}\setminus \{\textsf{cur}\}$ for some vertex $\textsf{cur}\notin \mathcal{P}[0]$, since we restrict that the initial vertex $\textsf{cur}$ is not in $\mathcal{P}[0]$ in line \ref{pull:judge}. Thus $\mathcal{P}[0]$ suffices that $\mathcal{P}[0]\subseteq \currentk{k'+1}$. 
    Second, consider the case where $\finalk{k'+1}\notin N(\mathcal{P}[0])$. All agents $i\in A$ which holds $\qfrom[i]\in \mathcal{P}[0]$ suffices $i\in \reservedk{k'+1}$ in line 11, then every vertex in $\mathcal{P}[0]$ never be chosen as $\textsf{cur}_{k'+1}$. Thus, it is clear that $\mathcal{P}[0]\subseteq \currentk{k'+1}$ by the same discussion as Lemma \ref{lem: connectivity}.
    
    There is a call of function \procname$(u_0,\emptyset, \mathcal{P}[0])$ for a vertex $u\in \nei{\mathcal{P}[0]}$. This call always succeeds, by lemma \ref{lem: firstcall}. Then it is clear that $p^{\max}_{\tau + 1}\ge |\mathcal{P}[0]|+1=p^{\max}_{\tau}+1$, which completes the proof.
\end{proof}
\subsection{Proof of \Cref{prop: runtime}}
\runtime*
\begin{proof}
    It takes $O(E(G[\qfrom]))=O(\Delta n)$ time for BFS, DFS in line \ref{pull:BFS}, \ref{pull:DFS}, and line \ref{pull:assign} takes $O(n)$ time. Thus, it takes $O(\Delta n)$ time for each call of procedure \procname. 

Line \ref{pull:sortcc} takes $O(\Delta n)$ time, and the number of calls is upper-bounded by $\Delta n$. Also, sorting the agent by distance takes $O(n)$ time by bucket sort by $O(n)$ bukcets, since
\begin{align*}
    \min_{(s',t)\in(\setqfrom\times T)}\dist(s',t)\le \min_{t\in T} \dist(s,t) \\\le \min_{(s',t)\in(\setqfrom\times T)}\dist(s',t) + n-1
\end{align*} for all $s \in \setqfrom$. Thus, it holds that
\begin{align*}\min_{(s',t)\in(\setqfrom\times T)}\dist(s',t)-1 \le \min_{t\in T} \dist(v,t)\\\le \min_{(s',t)\in(\setqfrom\times T)}\dist(s',t)+n 
\end{align*} for every $v\in \nei{\setqfrom}$. 
Then the total running time is $O(\Delta^2 n^2)$. 
\end{proof}

\subsection{Tightness, Adversarial instances}
\begin{prop}\label{prop: tight}
    There is an instance where the optimal makespan is $\textsf{diam}(G) + |A|-1$, even when the given graph is planar.
\end{prop}
\begin{proof}
    consider the graph, like the one illustrated in \Cref{fig: instance}(a). The optimal makespan is $k+\ell$, where $\textsf{diam}(G) + |A|-1 = (\ell +1) + k -1=\ell+k$. 
    We can observe that \algname returns a plan with makespan $\ell+k$, since the procedure \procname succeeds in determining $\current$ exactly once at each timestep.
\end{proof}
\begin{prop}\label{prop: adversarial}
    There is an instance where optimal makespan is 2, and \algname returns the plan with makespan $O(|A|)$.
\end{prop}
\begin{proof}
    consider the graph, like the one illustrated in \Cref{fig: instance}(b). Even if the optimal makespan is 2, \algname produces a solution with a makespan of at best $|A|/2+O(1)$.
\end{proof}

\section{Combining \algname and \lacam}\label{sec:lacam}
As \Cref{fig: instance}(b) illustrates, this algorithm performs poorly when  \emph{vertex connectivity} (the minimum number of vertices whose removal disconnects the graph) of $G[S]$ is small.
It is challenging to address this by improving the method used in \algname.
This is because \algname employs a method that gradually modifies the configuration within a single step while preserving connectivity; however, as seen in the instance in \Cref{fig: instance}(b), proceeding in that manner allows at most two agents to move forward for any $k$.
Therefore, we propose a search-based method that utilizes \lacam ~\cite{LaCAM23} towards an eventually optimal algorithm. 

\lacam is an eventually optimal \mapf solver that utilizes DFS over configuration space with some existing \mapf algorithms as a configuration generator. 
\lacam starts a DFS operation with a configuration $\conf{0}$, i.e. generating a reachable successor $\conf{1}$ from $\conf{0}$, generating $\conf{2}$ from $\conf{1}$, to the target configuraiton.
Unlike the normal DFS operation, \lacam allows to revisit a configuration: if $\conf{k}$ is revisited from $\conf{k'}$, \lacam introduces a constraint $C$ for a subset of agents during successor generation, such that $\conf{k}$ must produce a successor that is distinct from $\conf{k+1}$. 
Constraints $C$ are applied to configurations by specifying both which agents are constrained and their respective movements. 
Eventually, by enumerating all such constraints, the algorithm explores all possible neighboring configurations.
Thus, with search tree rewiring, the traversal explores the entire configuration space, implying completeness and optimality.

\lacam was originally developed for labeled \mapf, using PIBT~\cite{PIBT22} as a configuration generator, but it is easily attainable \prbname counterpart with \algname.
Configuration generator for \lacam needs to output the next configuration $\current$, taking the current configuration $\qfrom$, target $T$, and constraints $C$ as input. 
Here, the constraints $C$ include agents with predetermined destinations for the next step, along with their respective destinations.
Since \algname cannot handle the constraints, we modified it slightly to accommodate them (see \Cref{alg:PULL-LaCAM}). 
Note that these constraints may invalidate the lemmas in \Cref{sec:proof}, potentially preventing the attainment of a conflict-free configuration.
The adaptation for completeness is described as follows.

\begin{algorithm}[t]
		\caption{\algname for \lacam}
		\label{alg3}
		\begin{algorithmic}[1]  
			\Require $\qfrom$, $A$, $T$, and constraints $C$ \Comment{$C$ is given by set of tuples $(i\in A, v\in \neicl{\qfrom[i]})$}
			\Ensure configuration $\mathcal{Q}^{to}$, or $\bot $ \Comment{initialized by $\bot^n$}
            \State $R\leftarrow \emptyset$, $\textsf{pri}[v]\leftarrow \min_{t\in T}\dist(v,t)$ for $v\in \qfrom$
            \For{$(i,v)\in C$}
            \State $\current[i]=v$, $R\leftarrow R\cup \{i\}$
            \EndFor
            \If{$\exists i,j\in A$ s.t. $i\notin R\wedge j\in R\wedge \qfrom[i]=\current[j]$}
            \State $\textsf{pri}[\qfrom[i]]\leftarrow |V|$\label{pullcam:setpri}
            \EndIf
            \State perform \algname, with modified \procname
            \If{$\current$ is not connected} \Return $\bot$\label{pullcam:conn}
            \EndIf
            \If{there is a vertex conflict} \Return $\bot$\label{pullcam:vconf}
            \EndIf
            \If{there is a swap conflict between $i,j\in A$}
            \State $\current[i],\current[j]\leftarrow \current[j],\current[i]$\label{pullcam:swap}
            \EndIf
            \State\Return $\current$
            \Function{\procname}{$t,R, V'$} %
            \State perform lines \ref{pull:BFS} and \ref{pull:DFS} of \Cref{alg:PULL}
            \If{$V_S\leftarrow (\reachable\setminus \cutset) \setminus (V'\cup \{t\})= \emptyset$} \Return $R$ \label{pullcam:judge}
            \EndIf
            \State $\textsf{cur} \leftarrow \argmax_{v \in V_S} ( \textsf{pri}[v])$\label{pullcam:select}
            \State $i\leftarrow$ agents s.t. $\qfrom[i]=\textsf{cur}$
            \State $Q^{to}[i]\leftarrow \textsf{next}(\textsf{cur},t)$\Comment{determining $\current$}
            \State $R\leftarrow R\cup \{i\}$, $\textsf{cur} \leftarrow Q^{to}[i]$\label{pullcam:assign}
            \State \Return $R$
            \EndFunction
		\end{algorithmic}\label{alg:PULL-LaCAM}
	\end{algorithm}

A key change is in line \ref{pullcam:setpri}, which uses the list $\textsf{pri}$ instead of the distance to $T$.
If there is a pair of agents $i,j$ that $i\notin \reserved\land j\in \reserved\land\qfrom[i]=\current[j]$, we should move agent $i$ to avoid the vertex conflict.
Then the \Cref{alg:PULL-LaCAM} sets the priority of vertex $\qfrom[i]$ high, so that vertex $\qfrom[i]$ tends to be chosen as the initial vertex of the path. 
If $\qfrom[i]$ is chosen, agent $i$ moves to another vertex, which avoids the vertex-conflict between $i$ and $j$.
If there is are any conflicts in the resulting configuration, \Cref{alg:PULL-LaCAM} returns $\bot$.

If $\current$ is not a valid configuration: i.e. $\current$ is not connected, or there is any conflict, \Cref{alg:PULL-LaCAM} returns $\bot$, meaning it cannot generate the configuration in lines \ref{pullcam:conn} and \ref{pullcam:vconf}.
Note that the search process ensures all generated configurations are valid. 
Specifically: (1) if $\current$ is not connected, it is rejected in line \ref{pullcam:conn}; (2) any configuration with a vertex conflict is pruned by the search in line \ref{pullcam:vconf}; and (3) any swap conflicts are resolved in line \ref{pullcam:swap}. 
Consequently, every resulting configuration is guaranteed to be connected, collision-free, and reachable.
This leads to the following, which is a direct consequence of \cite{LaCAM23}.
\begin{prop}
    The algorithm, which uses \Cref{alg:PULL-LaCAM} as a subroutine of \lacam, is complete and eventually optimal; i.e., the output eventually converges to optima.
\end{prop}

\section{Experiments}
\subsection{detailed time profiles in \Cref{tab:ILP}}\label{sec:ILPapp}
In \Cref{tab:ILPdetail}, we present a detailed analysis of time and makespan. The 95\% confidence intervals are shown as $\pm$ next to the mean values.

\begin{table*}[t]
    \centering
\caption{Detailed results of \Cref{tab:ILP}. The columns ``time'' and ``ratio'' report the arithmetic mean together with the 95\% CI. In the last row, we report only the mean for the \ilp algorithm, since the number of solved instances is too small for the summary statistics to be meaningful.}\label{tab:ILPdetail}
\setlength{\tabcolsep}{8pt}
\renewcommand{\arraystretch}{1}
\begin{adjustbox}{max width=\textwidth}

\begin{tabular}{crrrrrr}
\toprule
&&&\multicolumn{2}{c}{time(\SI{}{\second})}&\multicolumn{2}{c}{ratio}\\
\cmidrule{4-7}
  Map & $n$ & solved (\%) & \ilp & \algname & speedup & makespan\\
\midrule
  \makecell[c]{{\scriptsize\emph{empty-8-8}}\\ \adjustbox{raise=-5mm}{\includegraphics[width = 0.08\linewidth]{figs/empty-8-8.png}}} 
& \makecell[r]{10\\15\\20\\25\\30}&\makecell[r]{100\\100\\100\\100\\100}&\makecell[r]{\mpar{0.137}{0.025}\\\mpar{0.226}{0.057}\\\mpar{0.233}{0.040}\\\mpar{0.202}{0.036}\\\mpar{0.190}{0.026}}  &\makecell[r]{\mpar{0.005}{0.001}\\\mpar{0.010}{0.002}\\\mpar{0.013}{0.002}\\\mpar{0.015}{0.002}\\\mpar{0.018}{0.003}} &
\makecell[r]{27.36\\18.73\\15.15\\11.95\\9.934}&
\makecell[r]{\mpar{1.347}{0.070}\\\mpar{1.484}{0.079}\\\mpar{1.540}{0.060}\\\mpar{1.678}{0.081}\\\mpar{1.713}{0.085}}\\
\midrule
  \makecell[c]{{\footnotesize\emph{random-32-32-20}}\\\adjustbox{raise=-5mm}{\includegraphics[width = 0.08\linewidth]{figs/random-32-32-20.png}}} &
\makecell[r]{10\\20\\30\\40\\50}&\makecell[r]{82\\48\\34\\24\\4}  &
\makecell[r]{\mpar{35.84}{17.27}\\\mpar{74.84}{38.93}\\\mpar{56.72}{37.89}\\\mpar{77.89}{48.05}\\93.66} &\makecell[r]{\mpar{0.018}{0.003}\\\mpar{0.041}{0.008}\\\mpar{0.061}{0.012}\\\mpar{0.092}{0.009}\\ \mpar{0.231}{0.022}} &
\makecell[r]{1660\\1499\\862.1\\804.1\\464.0}&
\makecell[r]{\mpar{1.131}{0.045}\\\mpar{1.460}{0.097}\\\mpar{1.650}{0.170}\\\mpar{1.676}{0.132}\\3.083}\\
\bottomrule
\end{tabular}
\end{adjustbox}
\end{table*}

\subsection{Table of results}\label{app:resulttable}
\Cref{tab:64-64,tab:48-48,tab:32-32,tab:warehouse,tab:mapsize} present the detailed average values and interquartile ranges for each experiment. The interquartile range is presented in parentheses.

\subsection{\lacam for adversarial instances}\label{sec:explacam}

Finally, we validate our eventually optimal solver on three adversarial instances, with their results shown in \Cref{fig:lacam-small}.
Furthermore,  \Cref{fig:lacam-iteration} illustrates the trajectory of makespan for small instances, with the number of iterations on the x-axis.
We can observe that \lacam obtains an optimal solution for sufficiently small cases. 
However, instances such as \emph{10-3-0} reveal \lacam's struggle to produce an optimal solution within two hours, indicating a need for more time to explore the entire configuration space.
We also tested larger maps, such as \emph{random-32-32-20} with randomly-generated instance ($|A|=100$) with $\text{LB}=22$. 
\Cref{fig:lacam-32-32-20} shows the results of the refinement. We can observe that there is a slight change in the solution quality. 
We tested other instances on \emph{random-32-32-20}, but observed the same trend without notable improvements in makespan.
This indicates room for further improvement.

\begin{figure}[t]

    \centering
    \begin{tikzpicture}

        \node[anchor=north west, inner sep=0] (image) at (0.15,5.8) {
        \includegraphics[width=0.9\linewidth]{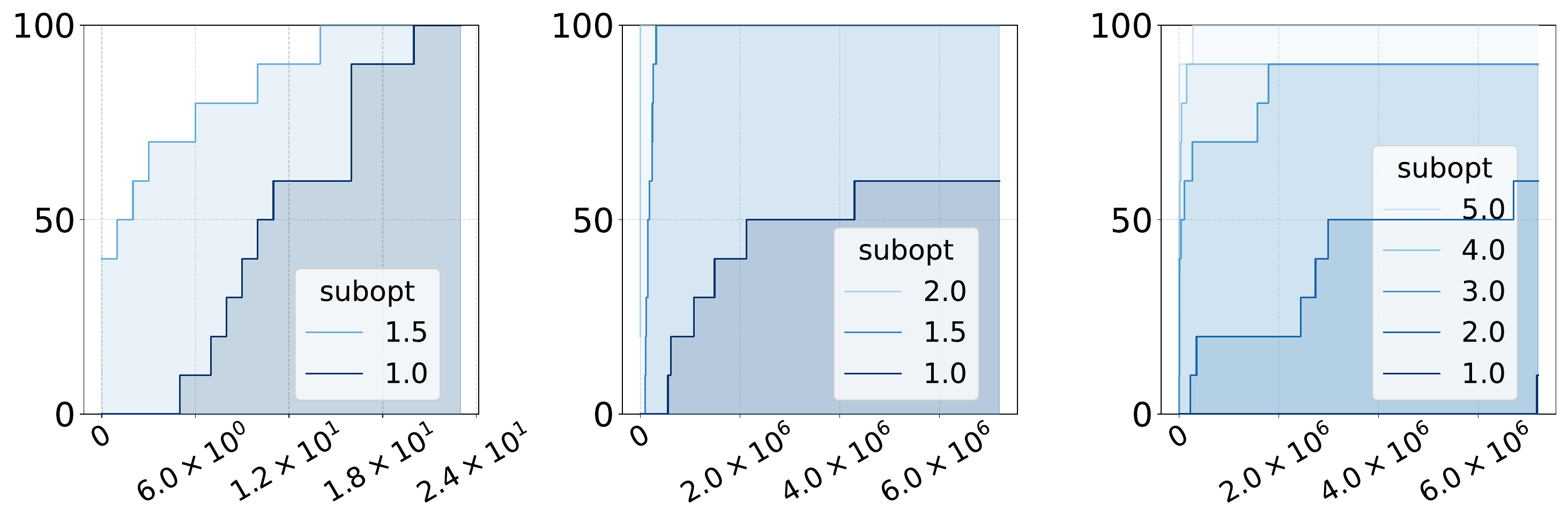}
    };
    \node[anchor = north west, inner sep = 0, rotate=90] at (0,4.1) {trials (\%)};
    \node[fill = white,anchor=north west, inner sep=0,rotate =90] (image) at (1.55,5.7) {
        \includegraphics[width=0.05\linewidth]{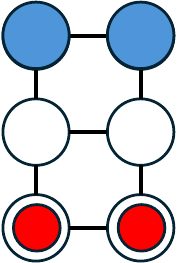}
    };
    \node at (1,5.9) {\emph{{\small2-3-0}}};
    \node[fill = white,anchor=north west, inner sep=0] (image) at (4.1,6.1) {
        \includegraphics[width=0.1\linewidth]{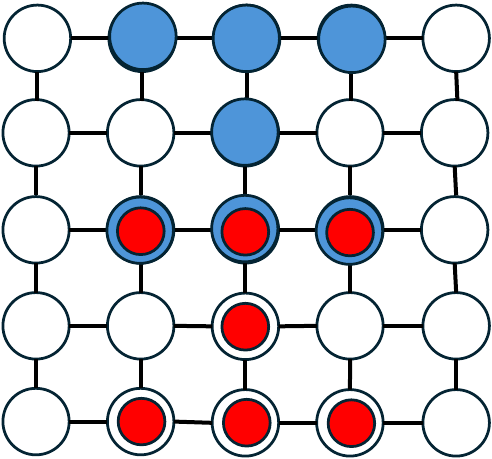}
    };
    \node at (3.5,5.9) {\emph{{\small5-5-0}}};
    \node at (5.9,5.9) {\emph{{\small10-3-0}}};
    \node[fill = white,anchor=north west, inner sep=0] (image) at (6.5,6.1) {
        \includegraphics[width=0.15\linewidth]{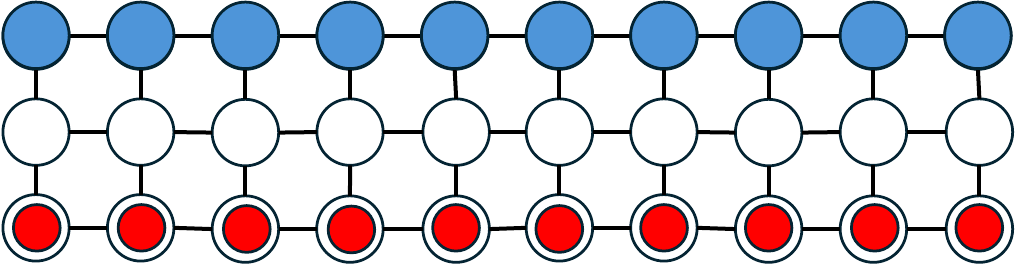}
    };
    \node[anchor = north west, inner sep = 0] at (2.5,3.3) {running time (ms)};
    \end{tikzpicture}
\caption{Experimental results of an eventually optimal algorithm using \lacam. The results are obtained from 10 runs with different random seeds. Each plot represents the percentage of trials that reached a specific makespan/LB (subopt).}
    \label{fig:lacam-small}
\end{figure}

\begin{figure}[ht]

    \centering
    \begin{tikzpicture}

        \node[anchor=north west, inner sep=0] (image) at (0.15,5.8) {
        \includegraphics[width=0.9\linewidth]{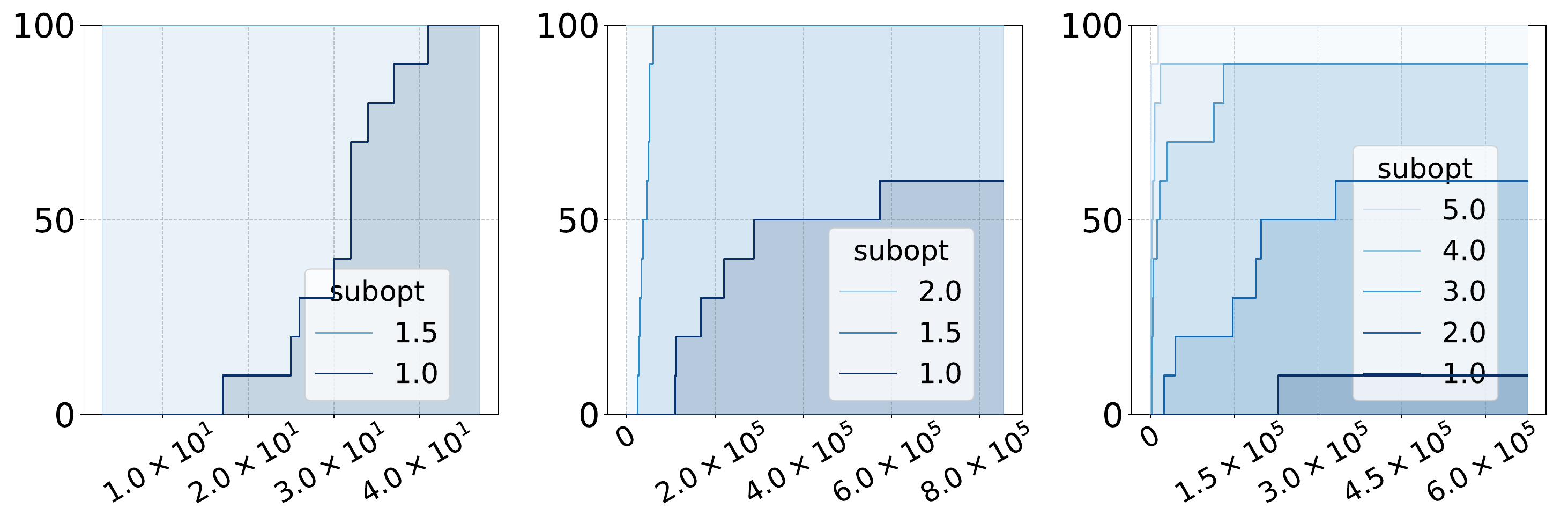}
    };
    \node[anchor = north west, inner sep = 0, rotate=90] at (0,4.1) {rate (\%)};
    \node at (1,5.9) {{\small2-3-0}};
    \node at (3.5,5.9) {{\small5-5-0}};
    \node at (5.9,5.9) {{\small10-3-0}};
    \node[anchor = north west, inner sep = 0] at (2.7,3.3) {number of iteration};
    \end{tikzpicture}

\caption{Experimental results of an eventually optimal algorithm using \lacam. The results are obtained from 10 runs with different random seeds. Each plot represents the percentage of trials that reached a specific makespan/LB.}
    \label{fig:lacam-iteration}
\end{figure}


\begin{figure}[ht]

    \centering
    \begin{tikzpicture}

        \node[anchor=north west, inner sep=0] (image) at (0.15,5.8) {
        \includegraphics[width=0.9\linewidth]{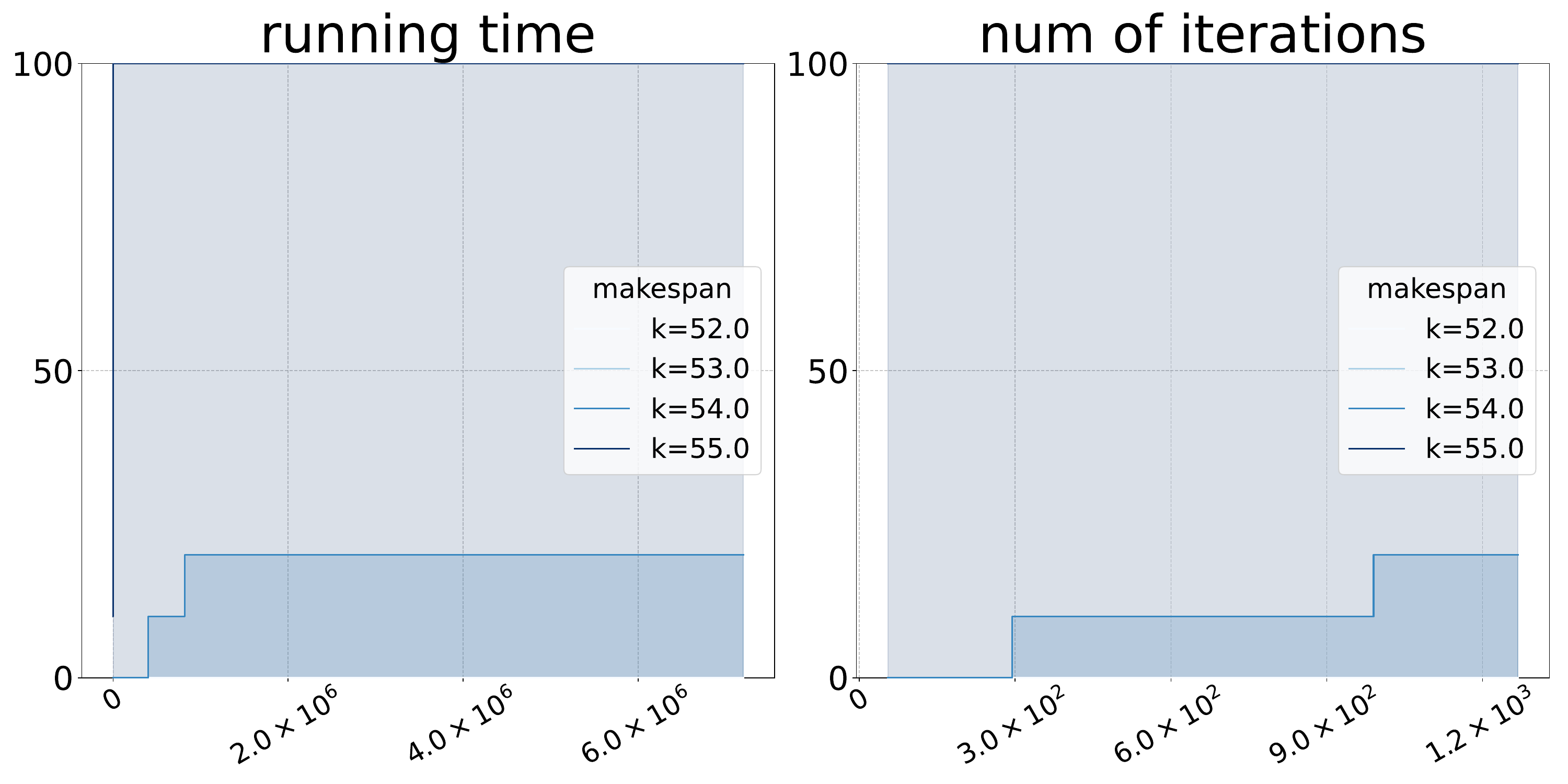}
    };
    \node[anchor = north west, inner sep = 0, rotate=90] at (0,3.4) {rate (\%)};
    \node[fill = white,anchor=north west, inner sep=0,rotate =90] (image) at (1.2,4.2) {
        \includegraphics[width=0.13\linewidth]{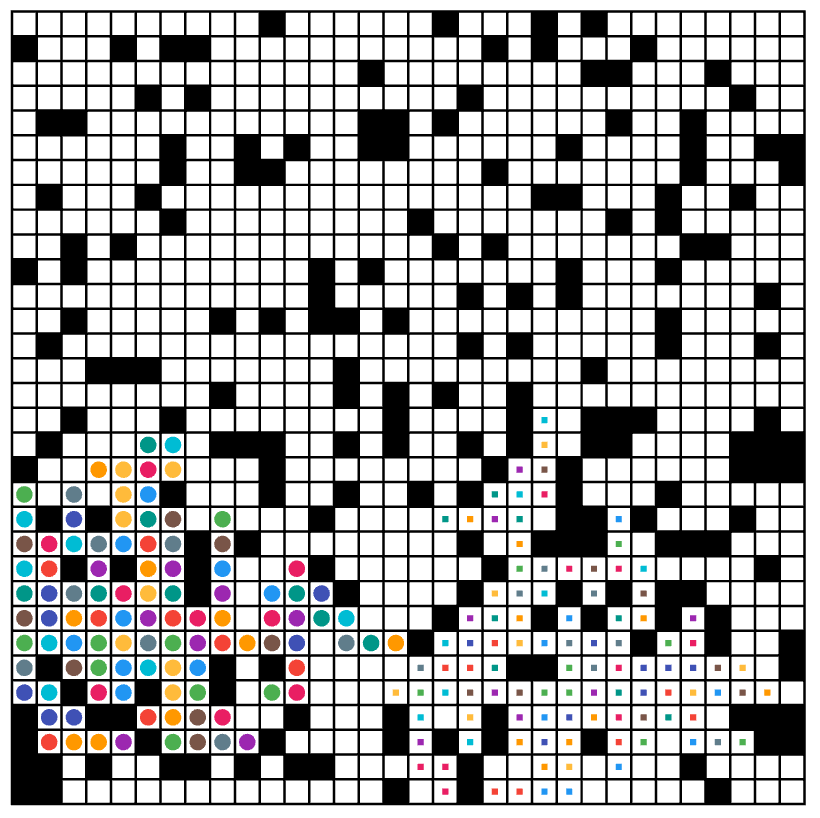}
    };
    \end{tikzpicture}

\caption{Experimental results of an eventualy optimal algorithm using \lacam, on \emph{random-32-32-20}. The results are obtained from 10 runs with different random seeds. }
    \label{fig:lacam-32-32-20}
\end{figure}

\begin{table*}[t]
\centering
\caption{random-64-64-20}\label{tab:64-64}
\setlength{\tabcolsep}{2.5pt}
\renewcommand{\arraystretch}{1.08}
\begin{adjustbox}{max width=\textwidth}
\begin{tabular}{@{}c l *{10}{c}@{}}
\toprule
 & & \multicolumn{10}{c}{$|A|$} \\
\cmidrule(lr){3-12}
\thead{metrics} & \thead{method}
 & 100 & 200 & 300 & 400 & 500 & 600 & 700 & 800 & 900 & 1000 \\
\midrule
time &
\mth{\algname}{single} &
\makecell{\valci{15.7}{13.5}{18.0}\\[-1pt]\valci{1.2}{1.1}{1.2}} &
\makecell{\valci{42.9}{36.3}{49.2}\\[-1pt]\valci{2.1}{2.0}{2.1}} &
\makecell{\valci{77.0}{64.5}{88.7}\\[-1pt]\valci{3.2}{3.1}{3.3}} &
\makecell{\valci{122.3}{101.6}{144.3}\\[-1pt]\valci{4.5}{4.4}{4.6}} &
\makecell{\valci{166.5}{143.6}{189.2}\\[-1pt]\valci{5.9}{5.7}{6.0}} &
\makecell{\valci{227.4}{200.4}{257.1}\\[-1pt]\valci{7.3}{7.2}{7.4}} &
\makecell{\valci{260.7}{226.2}{300.5}\\[-1pt]\valci{9.0}{8.9}{9.2}} &
\makecell{\valci{338.6}{299.4}{385.1}\\[-1pt]\valci{10.7}{10.5}{10.9}} &
\makecell{\valci{363.3}{317.9}{420.5}\\[-1pt]\valci{12.6}{12.4}{12.8}} &
\makecell{\valci{449.2}{397.3}{509.3}\\[-1pt]\valci{14.5}{14.3}{14.7}} \\
\midrule
\textsf{makespan}/LB &
\mth{\algname}{single} &
\makecell{\valci{2.102}{1.844}{2.261}\\[-1pt]\valci{3.100}{2.317}{3.681}} &
\makecell{\valci{2.445}{2.155}{2.684}\\[-1pt]\valci{5.439}{3.992}{7.057}} &
\makecell{\valci{2.629}{2.340}{2.908}\\[-1pt]\valci{7.238}{5.574}{8.723}} &
\makecell{\valci{2.774}{2.398}{3.082}\\[-1pt]\valci{8.615}{6.979}{10.453}} &
\makecell{\valci{3.056}{2.617}{3.362}\\[-1pt]\valci{11.013}{9.369}{12.832}} &
\makecell{\valci{2.997}{2.626}{3.370}\\[-1pt]\valci{12.022}{10.478}{14.199}} &
\makecell{\valci{3.208}{2.809}{3.542}\\[-1pt]\valci{13.598}{12.104}{15.079}} &
\makecell{\valci{3.202}{2.919}{3.433}\\[-1pt]\valci{15.028}{13.617}{16.496}} &
\makecell{\valci{3.500}{3.121}{3.773}\\[-1pt]\valci{16.380}{14.956}{18.277}} &
\makecell{\valci{3.414}{3.084}{3.639}\\[-1pt]\valci{17.537}{15.812}{19.305}} \\
\bottomrule
\end{tabular}
\end{adjustbox}
\end{table*}

\begin{table*}[t]
\centering
\caption{random-48-48-20}\label{tab:48-48}
\setlength{\tabcolsep}{2.5pt}
\renewcommand{\arraystretch}{1.08}
\begin{adjustbox}{max width=\textwidth}
\begin{tabular}{@{}c l *{10}{c}@{}}
\toprule
 & & \multicolumn{10}{c}{$|A|$} \\
\cmidrule(lr){3-12}
\thead{metrics} & \thead{method}
 & 100 & 200 & 300 & 400 & 500 & 600 & 700 & 800 & 900 & 1000 \\
\midrule
time &
\mth{\algname}{single} &
\makecell{\valci{15.2}{13.0}{17.2}\\[-1pt]\valci{1.1}{1.0}{1.1}} &
\makecell{\valci{38.3}{34.5}{43.3}\\[-1pt]\valci{2.1}{2.0}{2.1}} &
\makecell{\valci{72.8}{65.3}{81.7}\\[-1pt]\valci{3.2}{3.1}{3.2}} &
\makecell{\valci{110.1}{99.2}{123.0}\\[-1pt]\valci{4.5}{4.3}{4.5}} &
\makecell{\valci{149.6}{138.2}{172.5}\\[-1pt]\valci{5.9}{5.8}{6.1}} &
\makecell{\valci{182.2}{160.5}{204.6}\\[-1pt]\valci{7.5}{7.3}{7.6}} &
\makecell{\valci{226.1}{198.6}{260.5}\\[-1pt]\valci{9.2}{9.0}{9.3}} &
\makecell{\valci{261.4}{226.8}{308.5}\\[-1pt]\valci{11.0}{10.8}{11.1}} &
\makecell{\valci{320.8}{273.1}{372.0}\\[-1pt]\valci{12.7}{12.6}{12.9}} &
\makecell{\valci{361.6}{302.5}{415.2}\\[-1pt]\valci{14.6}{14.4}{14.8}} \\
\midrule
\textsf{makespan}/LB &
\mth{\algname}{single} &
\makecell{\valci{2.159}{1.907}{2.333}\\[-1pt]\valci{3.511}{2.595}{4.271}} &
\makecell{\valci{2.530}{2.242}{2.745}\\[-1pt]\valci{6.243}{5.233}{7.241}} &
\makecell{\valci{2.826}{2.409}{3.056}\\[-1pt]\valci{8.268}{7.102}{9.505}} &
\makecell{\valci{3.148}{2.789}{3.431}\\[-1pt]\valci{10.478}{9.621}{11.403}} &
\makecell{\valci{3.222}{2.794}{3.587}\\[-1pt]\valci{12.209}{10.985}{13.510}} &
\makecell{\valci{3.203}{2.795}{3.480}\\[-1pt]\valci{13.203}{11.884}{14.840}} &
\makecell{\valci{3.370}{3.069}{3.648}\\[-1pt]\valci{14.505}{13.330}{16.336}} &
\makecell{\valci{3.382}{3.111}{3.654}\\[-1pt]\valci{14.410}{12.360}{16.683}} &
\makecell{\valci{3.464}{3.151}{3.730}\\[-1pt]\valci{15.625}{13.000}{17.916}} &
\makecell{\valci{3.547}{3.180}{3.896}\\[-1pt]\valci{16.355}{14.414}{19.040}} \\
\bottomrule
\end{tabular}
\end{adjustbox}
\end{table*}

\begin{table*}[t]
\centering
\caption{random-32-32-20}\label{tab:32-32}
\setlength{\tabcolsep}{3pt}
\renewcommand{\arraystretch}{1.08}
\begin{adjustbox}{max width=\textwidth}
\begin{tabular}{@{}c l *{5}{c}@{}}
\toprule
 & & \multicolumn{5}{c}{$|A|$} \\
\cmidrule(lr){3-7}
\thead{metrics} & \thead{method}
 & 100 & 200 & 300 & 400 & 500 \\
\midrule
time &
\mth{\algname}{single} &
\makecell{\valci{12.2}{10.9}{13.6}\\[-1pt]\valci{1.1}{1.0}{1.1}} &
\makecell{\valci{31.9}{27.9}{35.4}\\[-1pt]\valci{2.1}{2.0}{2.1}} &
\makecell{\valci{55.1}{48.3}{62.5}\\[-1pt]\valci{3.2}{3.1}{3.3}} &
\makecell{\valci{76.6}{65.0}{87.7}\\[-1pt]\valci{4.6}{4.4}{4.7}} &
\makecell{\valci{94.3}{77.1}{108.6}\\[-1pt]\valci{6.0}{5.8}{6.1}} \\
\midrule
\textsf{makespan}/LB &
\mth{\algname}{single} &
\makecell{\valci{2.284}{1.892}{2.500}\\[-1pt]\valci{4.371}{3.805}{5.000}} &
\makecell{\valci{2.592}{2.282}{2.775}\\[-1pt]\valci{7.384}{6.709}{8.195}} &
\makecell{\valci{2.832}{2.574}{3.108}\\[-1pt]\valci{9.400}{8.424}{10.535}} &
\makecell{\valci{2.856}{2.579}{3.167}\\[-1pt]\valci{10.140}{8.475}{11.601}} &
\makecell{\valci{2.862}{2.563}{3.143}\\[-1pt]\valci{10.059}{8.122}{12.000}} \\
\bottomrule
\end{tabular}
\end{adjustbox}
\end{table*}

\begin{table*}[t]
\centering
\caption{warehouse-10-20-10-2-2}\label{tab:warehouse}
\setlength{\tabcolsep}{2.5pt}
\renewcommand{\arraystretch}{1.08}
\begin{adjustbox}{max width=\textwidth}
\begin{tabular}{@{}c l *{10}{c}@{}}
\toprule
 & & \multicolumn{10}{c}{$|A|$} \\
\cmidrule(lr){3-12}
\thead{metrics} & \thead{method}
 & 100 & 200 & 300 & 400 & 500 & 600 & 700 & 800 & 900 & 1000 \\
\midrule
time &
\mth{\algname}{single} &
\makecell{\valci{19.0}{15.3}{23.3}\\[-1pt]\valci{1.2}{1.1}{1.2}} &
\makecell{\valci{48.0}{39.1}{54.4}\\[-1pt]\valci{2.2}{2.1}{2.2}} &
\makecell{\valci{90.9}{69.8}{110.1}\\[-1pt]\valci{3.3}{3.2}{3.4}} &
\makecell{\valci{131.9}{108.8}{149.2}\\[-1pt]\valci{4.6}{4.4}{4.7}} &
\makecell{\valci{182.1}{144.0}{212.0}\\[-1pt]\valci{6.1}{5.9}{6.2}} &
\makecell{\valci{250.9}{202.0}{288.2}\\[-1pt]\valci{7.5}{7.3}{7.6}} &
\makecell{\valci{308.7}{250.4}{365.8}\\[-1pt]\valci{9.3}{9.0}{9.5}} &
\makecell{\valci{385.6}{321.8}{444.8}\\[-1pt]\valci{11.1}{10.7}{11.3}} &
\makecell{\valci{435.1}{353.3}{506.2}\\[-1pt]\valci{13.0}{12.6}{13.4}} &
\makecell{\valci{496.9}{432.8}{547.7}\\[-1pt]\valci{15.1}{14.5}{15.5}} \\
\midrule
\textsf{makespan}/LB &
\mth{\algname}{single} &
\makecell{\valci{1.739}{1.488}{1.953}\\[-1pt]\valci{2.660}{1.911}{2.967}} &
\makecell{\valci{2.102}{1.671}{2.388}\\[-1pt]\valci{4.478}{3.130}{5.520}} &
\makecell{\valci{2.145}{1.769}{2.365}\\[-1pt]\valci{6.378}{4.251}{8.100}} &
\makecell{\valci{2.404}{1.921}{2.692}\\[-1pt]\valci{8.263}{5.809}{10.265}} &
\makecell{\valci{2.441}{1.989}{2.595}\\[-1pt]\valci{9.750}{7.083}{12.395}} &
\makecell{\valci{2.394}{1.967}{2.566}\\[-1pt]\valci{10.310}{7.527}{12.352}} &
\makecell{\valci{2.436}{2.013}{2.735}\\[-1pt]\valci{11.616}{9.031}{13.637}} &
\makecell{\valci{2.635}{2.103}{2.985}\\[-1pt]\valci{13.244}{9.439}{15.504}} &
\makecell{\valci{2.718}{2.098}{3.252}\\[-1pt]\valci{14.091}{11.802}{16.729}} &
\makecell{\valci{2.748}{2.230}{3.175}\\[-1pt]\valci{15.227}{11.837}{18.048}} \\
\bottomrule
\end{tabular}
\end{adjustbox}
\end{table*}

\begin{table*}[t]
\centering
\caption{$|A|=100$, 20\% obstacles}\label{tab:mapsize}
\setlength{\tabcolsep}{3pt}
\renewcommand{\arraystretch}{1.08}
\begin{adjustbox}{max width=\textwidth}
\begin{tabular}{@{}c l *{7}{c}@{}}
\toprule
 & & \multicolumn{7}{c}{map size} \\
\cmidrule(lr){3-9}
\thead{metrics} & \thead{method}
 & 16$\times$16 & 24$\times$24 & 32$\times$32 & 40$\times$40 & 48$\times$48 & 56$\times$56 & 64$\times$64 \\
\midrule
time &
\mth{\algname}{single} &
\makecell{\valci{9.1}{7.7}{9.9}\\[-1pt]\valci{1.3}{1.1}{1.4}} &
\makecell{\valci{10.9}{9.6}{12.5}\\[-1pt]\valci{1.3}{1.1}{1.4}} &
\makecell{\valci{12.2}{10.9}{13.6}\\[-1pt]\valci{1.1}{1.0}{1.1}} &
\makecell{\valci{13.6}{11.5}{15.4}\\[-1pt]\valci{1.1}{1.0}{1.2}} &
\makecell{\valci{15.2}{13.0}{17.2}\\[-1pt]\valci{1.1}{1.0}{1.1}} &
\makecell{\valci{16.8}{14.7}{18.5}\\[-1pt]\valci{1.1}{1.0}{1.1}} &
\makecell{\valci{15.7}{13.5}{18.0}\\[-1pt]\valci{1.2}{1.1}{1.2}} \\
\midrule
\textsf{makespan}/LB &
\mth{\algname}{single} &
\makecell{\valci{2.19}{1.83}{2.37}\\[-1pt]\valci{5.15}{4.44}{5.95}} &
\makecell{\valci{2.44}{2.00}{2.69}\\[-1pt]\valci{4.81}{4.30}{5.32}} &
\makecell{\valci{2.28}{1.89}{2.50}\\[-1pt]\valci{4.37}{3.81}{5.00}} &
\makecell{\valci{2.26}{1.99}{2.45}\\[-1pt]\valci{4.04}{3.14}{4.71}} &
\makecell{\valci{2.16}{1.91}{2.33}\\[-1pt]\valci{3.51}{2.60}{4.27}} &
\makecell{\valci{2.05}{1.79}{2.24}\\[-1pt]\valci{3.17}{2.46}{3.79}} &
\makecell{\valci{2.10}{1.84}{2.26}\\[-1pt]\valci{3.10}{2.32}{3.68}} \\
\bottomrule
\end{tabular}
\end{adjustbox}
\end{table*}

\end{document}